\documentclass{article}
\usepackage[letterpaper]{geometry}
\usepackage{titlesec}

\titleformat{\section}
 {\normalfont\Large\sc\bfseries}
 {\thesection}{1em}{}
 
\titleformat{\subsection}
 {\normalfont\large\sc\bfseries}
 {\thesubsection}{1em}{}

\usepackage{dsfont} 
\usepackage{mathpazo}
\usepackage{MnSymbol}
\usepackage{amsmath}
\usepackage{mathrsfs}
 
\usepackage{amsthm}
\theoremstyle{plain}
\newtheorem{assumption}{Assumption}

\newtheorem{lemma}{\textbf{Lemma}}
\newtheorem{theorem}{\textbf{Theorem}}

\newtheorem{proposition}{\textbf{Proposition}}

\numberwithin{lemma}{section} 
\numberwithin{theorem}{section} 
\numberwithin{proposition}{section}
 
\usepackage{color}

\usepackage[dvipsnames]{xcolor}
\definecolor{brown}{RGB}{160,82,45}
\definecolor{steelblue}{RGB}{70,130,180}
\definecolor{darkred}{RGB}{139,0,0}
\usepackage{hyperref}
\hypersetup{
colorlinks = true,
urlcolor   = black,
citecolor  = black,
linkcolor  = black}

\usepackage{longtable}
\usepackage{multirow}
\usepackage{rotating,threeparttable}
\usepackage{booktabs}

\usepackage{appendix}

\usepackage{cleveref}
\crefname{assumption}{assumption}{assumptions}
\creflabelformat{assumption}{#2#1#3}

\usepackage{natbib}

\usepackage{authblk}

\newcommand\acknowledge[1]{%
  \begingroup
  \renewcommand\thefootnote{}\footnote{#1}%
  \addtocounter{footnote}{-1}%
  \endgroup
}

\def\realine{\mathbb{R}}
\def\indep{\!\perp\!\!\!\perp}
\def\cond{\ \vert\ }
\def\Pr{\mathbb{P}}
\def\Exp{\mathbb{E}}

\def\sumi{\hbox{$\sum_{i=1}^n$}}

\def\1{\mathbb{1}}
\def\Indi{\mathbb{1}}

\title{Testing for unobserved heterogeneous treatment effects}
\author[]{Yu-Chin Hsu\thanks{ Department of Economics, Academia Sinica. Email: \href{mailto:}{ychsu@econ.sinica.edu}.}}
\author[]{Ta-Cheng Huang\thanks{ (corresponding author) Department of Economics, Texas A{\&}M University. Email: \href{mailto:}{tchuang5@tamu.edu}.}}
\author[]{Haiqing Xu\thanks{ Department of Economics, University of Texas at Austin. Email: \href{mailto:}{h.xu@austin.utexas.edu}}}
\affil{}
\date{\today}

\begin{document}
\begin{titlepage}
\maketitle
\thispagestyle{empty}

\begin{abstract}
\fontsize{10}{15pt}\selectfont 
Unobserved heterogeneous treatment effects have been emphasized in recent policy evaluation literature. In this paper, we extend \cite{LW_2014_JoE}'s testing method for  unobserved heterogeneous treatment effects by developing nonparametric tests under  the standard exogenous instrumental variable assumption and allowing for endogenous treatment.  Specifically, we propose Kolmogorov--Smirnov--type statistics that are consistent and simple to implement.  To illustrate, we apply the proposed test method with two empirical applications: treatment effects of job training program on earnings as well as the impact of fertility on family income. The null hypotheses, i.e., lack of unobserved heterogeneous treatment effects,  cannot be rejected at a 10\% significance level in the former case, but should be rejected at all usual significance levels in the latter. 



\end{abstract}
\acknowledge{Acknowledgment: The third author would like to dedicate this paper to the memory of Professor Halbert White. We thank Qi Li and Quang Vuong for useful comments. }

\end{titlepage}


\newpage
\setcounter{page}{1}
\section{Introduction}
\fontsize{12}{18pt}\selectfont 
Unobserved heterogeneous treatment effects  have been emphasized in recent policy evaluation literature. See e.g. \cite{heckman1997making,matzkin2003nonparametric,Chesher_2003_ECMA,Chesher_2005_ECMA,CH_2005_ECMA,imbens2009identification}, \cite{heckman2001policy, heckman2005structural}. Recently,   \cite{LW_2014_JoE} and \cite{STU_2014_ER} develop nonparametric tests for unobserved heterogenous treatment effects via testing for additive separability of the error term in the structural relationship. A key assumption in their approaches is to assume  that treatment is (conditional) independent of the error term. Motivated by \cite{LW_2014_JoE},  in this paper we propose nonparametric tests under  the standard exogenous instrumental variable assumption and allowing for endogenous treatment.

In this paper, we consider the following structure model
\[
Y = g(D, X, \epsilon)
\]
where $Y$ is the outcome variable of interest, $X$ is a vector of observed covariates, $D$ denotes the binary treatment status, and $\epsilon$ is an unobserved error term of general form. In particular, $\epsilon$ represents the unobserved individual heterogeneity and we allow for the correlation between $\epsilon$ and $D$. Such a structural relationship is nonseparable in $\epsilon$, which implies treatment effect from $D$ on $Y$ varies across individuals, even after we control for observed heterogeneity $X$. When there is no unobserved heterogeneous treatment effects, we show that the structural model can be represented by 
\[
Y= m(D,X) + \nu(X, \epsilon)
\]
for some measurable functions $m$ and $\nu$. With additive separability of the error term, treatment effects are the same across individuals with the same covariates.  Formally, we test such an additive separability of the structural model.

A key feature of our approach is to allow for the presence of treatment endogeneity. Due to the sample selection issue highlighted in \cite{heckman1979sample},  the treatment status $D$ and error term $\epsilon$  are statistically dependent on each other given $X$. With additive separability, identification and estimation of average treatment effects directly obtain by \cite{IA_1994_ECMA}'s  \textit{``Local Average Treatment Effects''} and \cite{heckman1999local}'s  {\it Marginal Treatment effect} (MTE). For instance, \cite{angrist1991does} use a two--stage least square approach to estimate treatment effects in a linear specification.  
As is pointed out by \cite{heckman1997making}, however, the conventional assumption of identical treatment effects across individuals, while convenient, is implausible. If so, then the usual linear parametric specification with additive errors is not only for  feasibility and simplification of estimation, but also essential for identification as well as interpretations of treatment effects.


Though unobserved heterogeneity is crucially important, there are only a handful of papers studying testing for the existence of. 
In a framework without endogeneity, \cite{heckman1997making} focus on observed heterogeneous treatment effects. In particular, they provide tests for whether treatment effects depend on observed exogenous covariates via testing  zero--variance of $g(1,X,\epsilon)-g(0,X,\epsilon)$. Recently,  \cite{Hoderlein2009on} briefly discuss specification tests for unobserved heterogeneity in structured nonseparable models. This paper is intrinsically motivated by \cite{LW_2014_JoE} and \cite{STU_2014_ER}, who study nonparametric testing for unobserved heterogeneous treatment effects under the unconfoundedness assumption. In particular,   \cite{LW_2014_JoE} test such a hypotheses via testing an equivalent independence condition on observables under additional weak assumptions. Our paper extends \cite{LW_2014_JoE} by representing the existence of unobserved heterogeneous treatment effects by a similar conditional independence restriction on observables, allowing for the endogeneity of treatment variable.


Another closely related paper is \cite{HSU_2010_JoE} who study testing for the absence of selection on the gain to treatment in the generalized Roy model framework,  allowing for unobserved heterogenous treatment effect. This paper complements to  \cite{HSU_2010_JoE} in the sense that presence of both unobserved heterogeneity and selection into treatment is the so--called ``essential heterogeneity'' in \cite{heckman2006understanding}.

The proposed testing approach distinguishes the cases where exogenous  covariates contains or not a continuously distributed element. We propose Kolmogorov--Smirnov type statistics that are consistent and simple to implement. Motivated by \cite{stinchcombe1998consistent}, when there is a continuous covariate, we modify the classic Kolmogorov--Smirnov test statistics by using primitive functions of CDF to represent probability distributions of a (nonparametrically) generated variable.  Such a modification is novel and plays a key role for developing our test statistics. Moreover, we establish the asymptotic properties of the proposed test statistics under the null and alternative hypothesis.

The paper is organized as follows.  In Section 2, we introduce the framework and motivate our testing idea. \Cref{sec:condtional_indep} discusses our test statistics and main asymptotic results.  We distinguish the cases whether covariates include continuous variables. \Cref{sec:simulations} presents Monte Carlo experiments to study finite sample performance of our test statistics.   \Cref{sec:empirical} applies our testing approach to two empirical applications:  treatment effects of a job training program and treatment effects of fertility on earnings.
All proofs are collected  in the Appendix.

\section{Model and Testable Restrictions}
\label{sec:framework_motivation}
We consider the following  nonparametric nonseparable model:
\begin{equation}
\label{model}
Y=g(D,X,\epsilon)
\end{equation}
where $Y \in \mathbb{R}$, $D\in\{0,1\}$ and $X\in\mathbb R^{d_X}$ are observables,  $\epsilon\in\mathbb R^{d_\epsilon}$ is an unobserved random disturbance of general form, and $g$ is an unknown but smooth function defined on $\{0,1\}\times \mathscr S_{X\epsilon}$. In particular, $D$ is an endogenous treatment variable that is correlated with $\epsilon$ due to the selection issue. See e.g. \cite{heckman1997making}. 
To deal with the endogeneity issue, we follow the literature by introducing a binary instrumental variable $Z\in \{0,1\}$. Throughout the paper, we use upper case letters to denote random variables, and their corresponding lower case letters to stand for the realizations. 
Moreover, we use $\mathscr{S}_A$ for the support of a vector of generic random variables $A$. 

Note that the non-additivity of the structural relationship $g$ in $\epsilon$ captures the idea that individual treatment effect, i.e., $g(1,X,\epsilon)-g(0,X,\epsilon)$, depends on the unobserved individual heterogeneity $\epsilon$,  even after one controls for $X$ covariates.  Following \cite{LW_2014_JoE}, the null hypothesis for testing heterogeneous individual treatment effects  is equivalent to testing the following null hypothesis:  
\[
\mathbb{H}_{0}:\ g(D,X, \epsilon) = m(D,X) + \nu(X, \epsilon),
\] 
where $m: \mathscr{S}_{DX} \mapsto \realine$ and $\nu:\mathscr{S}_{X\epsilon} \mapsto \realine$.
Such an equivalence is summarized formally  in the following proposition.
\begin{proposition} 
\label{prop:equivalence bw hetero and separability}
Suppose  \eqref{model} holds, then $\mathbb{H}_0$ holds if and only if 
\begin{equation} 
\label{counterfactual}
g(1,X,\cdot) - g(0,X,\cdot)=\delta(X),
\end{equation}
holds for some measurable function $\delta(\cdot):\mathscr{S}_X \mapsto \mathbb{R}$.
\end{proposition}
\noindent
Note that under $\mathbb{H}_0$, $\delta(x) = m(1,x)-m(0,x)$, which represents homogenous individual treatment effects across individuals with the same value of covariates. \Cref{prop:equivalence bw hetero and separability}  shows that the additive separability of \eqref{model} is equivalent to  homogenous individual treatment effects. A similar result can be found in \cite{LW_2014_JoE}.




Another key insight from \cite{LW_2014_JoE} is the equivalence between the additive separability hypotheses and a conditional independence restriction on observables under the unconfoundedness assumption. 
Following \cite{LW_2014_JoE}, we derive a similar set of model restrictions in the presence of endogeneity.  For each $x\in\mathscr S_X$ and $z=0,1$, let $p(x,z) = \Pr (D=1|X=x,Z=z)$ be the {\it propensity score}.
\begin{assumption} 
\label{assmp:relevant IV}
Suppose $Z \indep \epsilon \cond X$ and $p(x,0) \neq p(x,1)$ for all $x \in \mathscr{S}_X$. 
Without loss of generality, let $p(x,0) < p(x,1)$ for all $x \in \mathscr{S}_X$.
\end{assumption} 

\begin{assumption}[Single--index error term]
\label{assmp:single index error}
There exist measurable functions $\tilde g: \mathscr{S}_{DX}\times \mathbb R\mapsto \mathbb{R}$ and $\nu:\mathscr{S}_{X\epsilon} \mapsto \realine$ such that 
\[
g(D,X,\epsilon) = \tilde g(D,X,\nu(X,\epsilon)).
\] Moreover, $\tilde g(d,x,\cdot)$  is strictly increasing in the scalar--valued index $\nu$ for $d=0,1$ and all $x \in \mathscr{S}_X$.
\end{assumption}
\noindent \Cref{assmp:relevant IV} requires the instrumental variable $Z$ to be (conditionally) exogenous and relevant,  which is standard in the literature.
See e.g., \cite{IA_1994_ECMA}, \cite{CH_2005_ECMA}, \cite{vuong2014counterfactual}, and references therein.   \Cref{assmp:single index error} imposes monotonicity on the structural relationship, which has also been widely assumed in the literature of nonseparable models.  For instance, \cite{matzkin2003nonparametric} \cite{Chesher_2003_ECMA} and \cite{CH_2005_ECMA} assume the structural function $g$ is strictly increasing in the scalar--valued error term $\epsilon$. 
It is worth pointing out that \Cref{assmp:single index error} holds under  $\mathbb{H}_0$. 

Under \Cref{assmp:relevant IV}, $\mathbb{H}_0$ implies that $\Exp(Y|X,Z=z) = \Exp\left[g(0,X,\epsilon)|X\right] + \delta(X) \times p(X,z)$ for $z=0,1$. Thus,  for each $x\in\mathscr S_X$, $\delta(x)$ can be identified by:
\begin{equation} 
\label{eqn:late}
\delta(x) \equiv \frac{\mu(x,1) - \mu(x,0)}{p(x,1)-p(x,0)},
\end{equation}
where $\mu(x,z) = \mathbb{E}(Y|X=x,Z=z)$ for $z=0,1$.
Note that $\delta(x)$ takes the conditioal version of LATE in \cite{IA_1994_ECMA}.

Let $W \equiv Y + (1-D) \times \delta(X)$. Given that \Cref{assmp:relevant IV,assmp:single index error} hold, by \Cref{prop:equivalence bw hetero and separability}, $\mathbb H_0$ implies that $W = g(1,X,\epsilon)$. We can also obtain $g(0,X,\epsilon)$ by a similar argument. Therefore,  by \Cref{assmp:relevant IV}, $W$ is conditionally independent of $Z$ given $X$. The next lemma summarizes above discussion.
\begin{lemma} 
\label{lemma:equivalence bw separability and conditional Indep}
Suppose \eqref{model}, and \Cref{assmp:relevant IV,assmp:single index error} hold. 
Then,  $\mathbb{H}_0$ implies that $W \indep Z \cond X$.
\end{lemma}
\noindent
\Cref{lemma:equivalence bw separability and conditional Indep} derives a testable model restriction under $\mathbb H_0$, i.e.,  $W \indep Z\ \vert\ X$. Note that $W$ can be consistently estimated by $Y + (1-D)\times\hat{\delta}(X)$, where $\hat{\delta}(X)$ is a consistent nonparametric estimator of $\delta(X)$. To provide a consistent test, we next provide weak conditions under which the conditional independence restriction is also sufficient for testing $\mathbb H_0$. 
\begin{assumption}[Monotone selection]
\label{assmp:monotone selection}
The selection to the treatment is given by
\begin{equation}
\label{selection}
D = \1 \left[ \theta(X, Z) - \eta \geq 0 \right], 
\end{equation} 
where $\theta$ is an unknown smooth function and $\eta \in \realine$ is an  unobserved error term.
\end{assumption}
\noindent
\cite{IA_1994_ECMA} first introduce the assumption that the selection to the treatment is monotone which implies the ``no defier'' condition. \cite{Vytlacil_2002_ECMA}  shows that such a monotonicity condition is observationally equivalent to the weak monotonicity of \eqref{selection} in the error term $\eta$. Furthermore, \cite{vuong2014counterfactual}  show that \Cref{assmp:monotone selection} can be relaxed to the strict monotonicity of $\mathbb P(Y\leq y;D=1|X,Z=1)-\mathbb P(Y\leq y;D=1|X,Z=0)$ in $y\in\mathscr S^\circ_{Y|X,D=1}$. 
 
Under \Cref{assmp:relevant IV}, note that $\theta(x,0) < \theta(x,1)$ for all $x\in\mathscr S_X$. Let $\mathcal C_x \equiv \{ \eta \in \mathbb{R}: \theta(x,0)<\eta \leq \theta(x,1)\}$ be the ``complier group'' given $X=x$ \citep[see][]{IA_1994_ECMA}.
\begin{assumption}
\label{assmp:sufficiently large compliers}
The support of $g(d,x,\epsilon)$ given $X=x$ and the complier group $\mathcal C_x$ equals to the support of $g(d,x,\epsilon)$ given $X=x$, i.e., 
\[
\mathscr{S}_{g(d,x,\epsilon) | X=x,\ \eta \in \mathcal C_x} = \mathscr{S}_{g(d,x,\epsilon)|X=x}.
\]
\end{assumption}

\noindent
\Cref{assmp:sufficiently large compliers} is a support condition. This assumption is testable, since the distribution of $g(d,x,\epsilon)$ given  $X=x$ and $\eta\in\mathcal C_x$ can be identified, see, e.g., \cite{IR_1997_ReStud}.
Specifically, for all $t\in\mathbb R$,
\[
F_{g(d,x,\epsilon)|X=x,\eta\in \mathcal C_x}(t)= \frac{\Pr(Y\leq t, D=d|X=x, Z=1)-\Pr(Y\leq t, D=d|X=x,Z=0)}{\Pr(D=d|X= x, Z=1)-\Pr(D=d|X=x, Z=0)},
\]
from which we can identify the support $\mathscr S_{g(d,x,\epsilon)|X=x,\eta\in\mathcal C_x}$. 

\begin{theorem} 
\label{thm:equivalence bw separability and conditinoal Indep}
Suppose \eqref{model}, and \Cref{assmp:relevant IV,assmp:monotone selection,assmp:single index error,assmp:sufficiently large compliers} hold.  
Then $\mathbb{H}_0$ holds if and only if  $W \!\perp\!\!\!\perp Z|X$.
\end{theorem}
\noindent
Throughout, we maintain \Cref{assmp:relevant IV,assmp:monotone selection,assmp:single index error,assmp:sufficiently large compliers}.  By \Cref{thm:equivalence bw separability and conditinoal Indep}, $\mathbb H_0$ can be tested by testing  the conditional independence condition in the theorem.  As a matter of fact, \Cref{thm:equivalence bw separability and conditinoal Indep} is the basis of our approach to test for (unobserved) heterogeneity in treatment effects.

\subsection{Discussion: testing for full additive separability}
One might be also interested in testing for the (full) additive separability of the error term in the outcome equation \eqref{model}, which has been widely used in the empirical treatment effect literature. \cite{LW_2014_JoE} and \cite{STU_2014_ER} consider testing additive error structure under the unconfoundedness assumption. Specifically,  their null hypothesis are given by
\[
\mathbb{H}^\dag_{0}:\ g(D,X, \epsilon) = m^\dag(D,X) + \epsilon
\] for some measurable function $m^{\dag}$. Clearly,  $\mathbb{H}^\dag_0$ is more restrictive than our null hypotheses $\mathbb H_0$. In particular, $\mathbb H^\dag_0$ rules out both observed and unobserved heterogeneity, which has been widely discussed in the treatment effect literature. See e.g. \cite{heckman2007econometric} and \cite{imbens2009recent}.

Following \cite{LW_2014_JoE}  and \Cref{thm:equivalence bw separability and conditinoal Indep}, we derive a similar set of conditional independence restrictions that are equivalent to $\mathbb{H}^\dag_0$.
\begin{theorem} 
\label{thm:heterogeneity and heteroscadesticity}
Suppose \Cref{assmp:relevant IV,assmp:monotone selection,assmp:single index error,assmp:sufficiently large compliers} hold. Suppose in addition $X$ and $\epsilon$ is independent, i.e.,  $X \indep \epsilon$. 
Then $\mathbb{H}^\dag_0$ holds if and only if 
\begin{equation}
\label{eqn:W* indep X,Z}
W-\Exp(W|X) \indep (X,Z).
\end{equation}
\end{theorem}
\noindent It is worth noting that  \eqref{eqn:W* indep X,Z} is equivalent to (i) $W \indep Z \cond X$; and (ii) $W-\Exp(W|X) \indep X$ \citep[see e.g.][]{dawid1979conditional}.
Condition (i) is the same as that in \Cref{thm:equivalence bw separability and conditinoal Indep} and (ii) has also been derived in \cite{LW_2014_JoE} for testing $\mathbb H^\dag_0$ under the unconfoundedness assumption. 

\section{Consistent test for unobserved treatment effect heterogeneity}
\label{sec:condtional_indep}
In this section, we propose tests for unobserved treatment effect heterogeneity via testing the conditional independence restriction, i.e., $W \indep Z \cond X$. Because $Z$ is binary, the conditional independence becomes $F_{W|X,Z}(w|x,0)=F_{W|X,Z}(w|x,1)$ for all $x\in\mathscr S_X$. Difficulties arise due to the fact that $W= Y + (1-D) \times {\delta}(X)$ needs to be nonparametrically estimated from the data, in particular when $X$ includes continuous variables.

In the following discussion, we distinguish the cases whether covariates $X$ include continuous variables.  For expositional simplicity, throughout we assume $X \in \realine$ be a scalar--valued random variable.  It is straightforward to generalize our result to vector-valued covariates. 
For the discrete case, our test adopts the classic two--sample Kolmogorov--Smirnov test.  When $X$ is continuous, we propose a modified Kolmogorov--Smirnov test that also converges to a limiting distribution at the $\sqrt n$--rate.

\subsection{Discrete Covariates}
Let $\{(Y_i,D_i,X_i, Z_i): i=1,\cdots,n\}$ be an i.i.d.\ random sample of $(Y,D,X,Z)$, where $X$ is distributed on a finite support. By \Cref{thm:equivalence bw separability and conditinoal Indep}, we can test  unobserved treatment effect heterogeneity by testing the model restriction 
\[
F_{W}(\cdot|X,Z=0) = F_{W}(\cdot|X,Z=1).
\]

For each $x\in\mathscr S_X$, let
\[
\hat{\delta}(x) = \frac{\hat{\mu}(x,1)-\hat{\mu}(x,0)}{\hat{p}(x,1)-\hat{p}(x,0)},
\]
where, for $z=0,1$,
\[
\hat{\mu}(x,z) = \frac{\sumi Y_i \1(X_i=x, Z_i=z)}{\sumi \1(X_i=x, Z_i=z)} \text{ and } \hat{p}(x,z) = \frac{\sumi D_i \1(X_i=x, Z_i=z)}{\sumi \1(X_i=x, Z_i=z)}.
\] 
It is straightforward that $\hat \delta(x)$ converges to $\delta(x)$ at the $\sqrt n$--rate under additional regularity conditions. Let $\hat W_i=Y_i+(1-D_i)\hat \delta(X_i)$. By definition, $\hat W_i-W_i=(1-D_i)\big[\hat \delta(X_i)-\delta(X_i)\big]$ is the first--stage estimation error.

We are now ready to define our test statistic as
\[
\hat{\mathcal{T}}_n = \sup_{w \in \mathbb{R};\ x \in \mathscr{S}_X} \sqrt{n} \left\vert \hat{F}_{\hat{W}|XZ}(w|x,0) - \hat{ F}_{\hat{W}|XZ}(w|x,1) \right\vert
\]
where $\hat{F}_{\widehat{W}|XZ}(w|x,z)$ is the empirical CDF of $\hat{W}$ conditional on $(X,Z) = (x,z)$, i.e., 
\[
\hat{F}_{\hat{W}|XZ}(w|x,z) = \frac{\sumi \1(\hat{W}_i \le w) \1(X_i=x, Z_i=z)}{\sumi \1(X_i=x, Z_i=z)}.
\] 

Next, we establish the limiting distribution of the proposed test statistic. For notational simplicity,  let $\1_{XZ}(x,z) = \1(X=x,Z=z)$ and $f_{WD|XZ}(w,d|x,z) \equiv f_{W|DXZ}(w|d,x,z) \times \Pr(D=d|X=x,Z=z)$. 
Let
\[
\kappa(w,x)=- \frac{f_{WD|XZ}(w,0|x,1)-f_{WD|XZ}(w, 0|x,0)}{p(x,1)-p(x,0)}.
\] 
It is worth noting that $\kappa(w,x) \geq 0$ can be interpreted as the p.d.f.\ of the potential outcome $g(1,X,\epsilon)$ given the complier group under \Cref{assmp:relevant IV,assmp:monotone selection}. 
Moreover, let 
\begin{align}
\label{eqn:dpsi}
\psi_{wx} &= [\1(W \le w) - F_{W|X}(w|x)] \times \left[ \frac{\1_{XZ}(x,0)}{\Pr(X=x,Z=0)} -  \frac{\1_{XZ}(x,1)}{ \Pr(X=x,Z=1)}  \right ];\\ 
\label{eqn:dphi}
\phi_{wx} &= \kappa(w,x)\times [ W - \Exp(W|X)]\times \left[ \frac{\1_{XZ}(x,0)}{ \Pr(X=x,Z=0)} -  \frac{\1_{XZ}(x,1)}{\Pr(X=x,Z=1)}  \right ].
\end{align}
By definition, $\psi_{wx}$ and $\phi_{wx}$ and are random objects indexed by $(w,x)$. In particular, $\mathbb E (\psi_{wx}|X,Z)=\mathbb E (\phi_{wx}|X,Z)=0$ under $\mathbb H_0$.

\begin{assumption}
\label{ass5}
Let  $X$ be a discrete random variable. Moreover, the probability distribution of $Y$ given $(D,X,Z)$ admits a uniformly continuous density function $f_{Y|DXZ}$.
\end{assumption}
\begin{theorem}
\label{thm:test_noX}
Suppose \Cref{assmp:relevant IV,assmp:monotone selection,assmp:single index error,assmp:sufficiently large compliers,ass5} hold.   Then, under $\mathbb H_0$, 
\[
\hat{\mathcal T}_n  \overset{d}{\rightarrow} \ \sup_{w \in \mathbb{R};  \ x\in\mathscr S_X}\  |\mathcal{Z}(w,x)|
\] 
where $\mathcal{Z}(\cdot,x)$ is a mean--zero Gaussian process with covariance kernel: for $(w,x), (w',x') \in \mathbb{R} \times \mathscr{S}_X$,
\[
\text{Cov} \left[ \mathcal{Z}(w,x), \mathcal{Z}(w',x') \right] = \Exp \left[ (\psi_{wx} + \phi_{wx}) (\psi_{w'x'} + \phi_{w'x'})\right].
\]
Moreover, under $\mathbb H_1$, we have
\[
n^{-\frac{1}{2}}\hat{\mathcal T}_n\overset{p}{\rightarrow} \sup_{w \in \mathbb{R};\ x \in \mathscr{S}_X} \left\vert {F}_{{W}|XZ}(w|x,0) - { F}_{{W}|XZ}(w|x,1) \right\vert.
\]
\end{theorem}
\noindent
\Cref{thm:test_noX}  forms the basis for the following one-sided test against any alternative to $\mathbb H_0$: reject $\mathbb H_0$ significance level $\alpha$ if and only if $\hat{\mathcal T}_n\geq c_{\alpha}$. Regarding the limiting distribution $\mathcal Z(\cdot,x)$, $\phi_{wx}$ in the covariance kernel appears due to the first--stage estimation of $\delta(x)$. 

Because the asymptotic distribution of $\hat{\mathcal T}_n $ under $\mathbb H_0$ is complicated and it it is computationally difficult to derive  the limiting distribution for the critical value, then we apply the multiplier bootstrap method in  \cite{van1996weak}, \cite{donald2003ecma}, and \cite{hsu2016multiplier} to approximate the entire process and then to approximate critical values.  Specifically, we simulate a sequence of i.i.d.\ pseudo random variables  $\{U_i: i=1,\cdots,n\}$ with $\Exp(U)=0$, $\Exp(U^2)=1$, and $\Exp(|U^4|) < +\infty$. Moreover,  the simulated sample $\{U_i: i=1,\cdots,n\}$ is independent of the random sample $\{(Y_i,X_i,D_i,Z_i): i=1,\cdots,n\}$. Then, we obtain the following simulated empirical process:
\[
\hat{\mathcal{Z}}^u(w,x) = \frac{1}{\sqrt{n}} \sum_{i=1}^n U_{i} \times ( \hat{\psi}_{wx,i} + \hat{\phi}_{wx,i}),
\]
where $\hat{\psi}_{wx} + \hat{\phi}_{wx}$ is the estimated influence function such that
\begin{align*}
\hat{\psi}_{wx} &=  [\Indi(\hat{W} \le w) - \hat{F}_{\hat{W}|X}(w|x)] \times   \left[ \frac{\mathbb 1_{XZ}(x,0)}{\hat{\Pr}(X=x,Z=0)} -  \frac{\mathbb 1_{XZ}(x,1)}{ \hat{\Pr}(X=x,Z=1)}  \right ];\\ 
\hat{\phi}_{wx} &= \hat{\kappa}(w,x) \times   \left[ \hat{W} -  \frac{\sumi \hat{W}_i \Indi(X_i=x)}{\sumi \Indi(X_i=x)}\right] \times \left[ \frac{\mathbb 1_{XZ}(x,0)}{ \hat{\Pr}(X=x,Z=0)} -  \frac{\mathbb 1_{XZ}(x,1)}{\hat{\Pr}(X=x,Z=1)}  \right ]
\end{align*}
where 
\begin{align*}
&\hat{\Pr}(X=x,Z = z) = \frac{1}{n} \sumi \Indi(X_i=x,Z_i=z),\ z=0,1,\ \text{ and} \\
&\hat{\kappa}(w,x) = -\frac{\hat{f}_{{W}D|XZ}(w,0|x,1)-\hat{f}_{{W}D|XZ}(w, 0|x,0)}{\hat{p}(x,1) - \hat{p}(x,0)},
\end{align*}
in which $\hat{f}_{{W}D|XZ}(w,0|x,z)$ and $\hat p(x,z)$ are nonparametric estimators of ${f}_{{W}D|XZ}(w,0|x,z)$ and $ p(x,z)$, respectively. 
See e.g. \cite{hsu2016multiplier} for more details.  
By a similar argument to \cite{donald2003ecma} and \cite{hsu2016multiplier},  $\hat{\mathcal{Z}}^u(\cdot,x) $ converges to the same limiting process $\mathcal{Z}(\cdot,x)$. Next, to derive the critical values, we first let $\Pr_U$ be the multiplier probability measure.
Then, for a given significant level $\alpha$, the simulated critical value $\hat{c}_{n}(\alpha)$ is defined as
\[
\hat{c}_{n}(\alpha) = \sup \left\{ q: \Pr_U \left( \sup_{w \in \realine,\ x \in \mathscr{S}_{X}} \left| \hat{\mathcal{Z}}^u(w,x)\right| \le q \right) \le 1-\alpha  \right\}.
\]
By definition, $\hat{c}_{n}(\alpha)$ is the $(1-\alpha)$ quantile of the simulated  distribution. 
With the simulated critical value, we reject $\mathbb H_0$ if and only if  $\hat{\mathcal{T}}_n > \hat{c}_n(\alpha)$.


\subsection{Continuous Covariates}
When $X$ contains continuous covariates, the generated regressor $\hat{W}$ involves estimating a nonparametric function $\delta(\cdot)$. Because $\hat W$ appears in the indicator function of Kolmogorov--Smirnov--type test statistics, then the empirical process argument in the proof of \Cref{thm:test_noX} does not apply to the continuous covariates case. Therefore, we propose a modified Kolmogorov--Smirnov test that we can derive its limiting distribution.

Let $\lambda(t)= -t\times \mathbb 1 (t\leq 0)$ and $\Pi(w|x,z)= \mathbb E [\lambda(W-w)|X=x, Z=z]$. As a matter of fact,  $\Pi(\cdot |x,z) $  is the primitive function of the $F_{W|XZ}(\cdot|x,z)$, i.e., $\frac{\partial }{\partial w} \Pi(w|x,z) = F_{W|XZ}(w|x,z)$. Then $\Pi(\cdot|X,Z)$ characterizes the (conditional) distribution of $W$ given $X$ and $Z$ the same as the c.d.f. $F_{W|XZ}$. 
Hence, $W \indep Z \cond X$ is equivalent to the following: 
\[
\mathbb{H}^{\pi}_{0}: \Pi(\cdot|x,0) =\Pi(\cdot|x,1), \ \ \forall \ x\in\mathscr S_X
\] 
Note that $\lambda(\cdot)$ is a differentiable function.
W.l.o.g., we assume $\mathscr{S}_W$ is bounded.\footnote{ When $\mathscr S_W$ is unbounded, we need to modify $\Pi(w|x,z)$ by $\tilde \Pi(w|x,z)= 2(w^2+C)]^{-1/2} \times \Pi(w|x,z)$ where $C\geq \mathbb E (W^2)$. The modification ensures $\tilde \Pi(\cdot|x,z)$ is uniformly bounded above. Then all our arguments remain valid.} 

Let  $G(w,x; z)=\mathbb E \left[\mathbb 1^*_{XZ}(x,z) q(X,z') \lambda (W-w)\right]$, where $z'=1-z$, $q(x,z)=f_{X|Z}(x|z)\mathbb P(Z=z)$ and  $\mathbb 1^*_{XZ}(x,z) = \mathbb 1 (X\leq x;Z=z)$. 
Following \cite{stinchcombe1998consistent}, we rewrite the conditional moments  by the following unconditional restrictions: 
\[
\mathbb{H}^G_0:  \ G(w,x;0)-G(w,x;1)=0, \ \ \forall (w,x)\in\mathbb R\times \mathscr S_X.
\]
To see the equivalence between $\mathbb{H}^{\pi}_{0}$ and $\mathbb{H}^{G}_{0}$, note that 
\[
\frac{\partial}{\partial x} \Exp[ \Indi(X \le x)   \lambda(W-w)  f_{X|Z}(X|1-z)| Z=z] = \Pi(w|x,z)  q(x,0)q(x,1).
\]
Thus, our test statistic is constructed based on $\mathbb{H}^G_0$.

Let $K$ and $h$ be a bounded kernel function and a smoothing bandwidth, respectively.  
By eq. \eqref{eqn:late}, we nonparametrically estimate $\delta(X_i)$ by 
\[
\hat{\delta}(X_i)=\frac{ \hat{\mu}(X_i,1)-\hat{\mu}(X_i,0)}{\hat{p}(X_i,1)-\hat{p}(X_i,0)}.
\] 
where, for $z=0,1$, 
\[
\hat{\mu}(X_i,z) = \frac{ \sum_{j\neq i}Y_j K \big( \frac{X_j-X_i}{h} \big)  \Indi (Z_j=z)}{\sum_{j\neq i}K \big( \frac{X_j-X_i}{h} \big)  \Indi (Z_j=z)}  \text{ and } \hat{p}(X_i,z) = \frac{ \sum_{j\neq i} D_j  K\left( \frac{X_j-X_i}{h} \right)  \Indi (Z_j=z)}{\sum_{j\neq i} K \big( \frac{X_j-X_i}{h} \big) \Indi (Z_j=z)}.
\]

Let $\hat{q}(X_i,z) =\frac{1}{(n-1) h} \sum_{j\neq i} K \big( \frac{X_j-X_i}{h} \big) \Indi(Z_j = z)$ be the estimator of $q(X_i,z)$.  Moreover, define 
\[
\hat{\mathcal{T}}^{c}_{n} = \sup_{(w,x) \in \mathscr{S}_{WX}}\ \sqrt{n} \left| \hat{G}(w, x; 0) - \hat{G}(w, x; 1) \right|
\]
where 
\[
\hat{G}(w,x; z) = \frac{1}{n}\sum_{i=1}^n \Indi^{*}_{X_iZ_i}(x,z) \times \hat{q}(X_i,1-z) \times\lambda (\hat{W}_i-w).
\] 
In above definition, the support $\mathscr S_{WX}$ is assumed to be known for simplicity. 
In practice, this assumption can be relaxed by using a consistent set estimator $\hat{\mathcal{S}}_{WX}$ of $\mathscr{S}_{WX}$.

As is shown below, the proposed test statistics $\hat{\mathcal T}^c_n$ converges in distribution to a limit at the regular $\sqrt n$ rate. 
The proofs proceed in two steps: we first show that $\hat{G}(w,x;z)$ can be approximate by 
\[
\tilde G(w,x;z)=\frac{1}{n}\sum_{i=1}^n \mathbb 1^*_{X_iZ_i}(x,z)\times  \hat q(X_i,1-z)\times(w - \hat{W}_i) \times \mathbb{1}(W_i \le w).
\]  
It is worth noting that $\tilde{G}(w,x;z)$ has no nonparametric estimator in the indicator function, which allows us to apply $\mathbb{U}$-processes theorem to establish its limiting distribution in the second step.

The first step of our proof is to show that for $z=0,1$
\[
\sup_{(w,x)\in\mathscr S_{WX}}\ \left| \hat{G} (w,x;z)- \tilde{G}(w,x;z)\right| = o_p \left( n^{-1/2} \right).
\] 
A key condition for above result is that the nonparametric estimator $\hat{\delta}(\cdot)$ converges to $\delta(\cdot)$ uniformly at a rate faster than $n^{-\iota}$ for some $\iota>\frac{1}{4}$.
To show the approximation, we assume the following regularity assumptions.
\begin{assumption} 
\label{assmp:densitysmoothness}
 $\mathscr S_W\subseteq \mathbb R$ is a compact subset. Let  $ \sup_{(x,z)\in\mathscr S_{XZ}} f_{X|Z}(x|z)\leq \overline f$ for some $\overline f<+\infty$. Moreover, $\inf_{x\in\mathscr S_X} |q(x,1)-q(x,0)|>0$.
\end{assumption}

\begin{assumption} 
\label{assmp:bandwidth}
For some $\iota>\frac{1}{4}$, $h\rightarrow 0$ and $ n^{\iota}\frac{1}{\sqrt{n h}}\rightarrow 0$  as $n\rightarrow \infty$.
\end{assumption}

\begin{assumption}
\label{assmp:kernel}
The first stage nonparametric estimators satisfy:
\begin{align*}
&\sup_{(x,z)\in\mathscr S_{XZ}}\Big|\mathbb E \Big[\frac{1}{nh}\sum_{j=1}^n K \big(\frac{X_j-x}{h}\big) \mathbb 1(Z_j=z)\Big]-q(x,z)\Big|= O_p(n^{-\iota}),\\
&\sup_{(x,z)\in\mathscr S_{XZ}}\Big|\mathbb E \Big[\frac{1}{nh}\sum_{j=1}^n D_j K \big(\frac{X_j-x}{h}\big) \mathbb 1(Z_j=z)\Big]-p(x,z) q(x,z)\Big|= O_p(n^{-\iota}),\\
&\sup_{(x,z)\in\mathscr S_{XZ}}\Big|\mathbb E\Big[\frac{1}{nh}\sum_{j=1}^n Y_j K \big(\frac{X_j-x}{h}\big) \mathbb 1(Z_j=z)\Big]-\mathbb E (Y|X=x,Z=z) q(x,z)\Big|= O_p(n^{-\iota}),
\end{align*} 
\end{assumption}
\noindent
\Cref{assmp:densitysmoothness} is weak and standard in the literature. 
\Cref{assmp:bandwidth,assmp:kernel} require the standard deviation and bias term of nonparametric estimation converge to zero uniformly at a rate no slower than $n^{-\iota}$, respectively.  
In particular, \Cref{assmp:kernel} is a high-level assumption that can be derived by primitive conditions on $K$ and $h$.

\begin{lemma}
\label{lem:approximation}
Suppose \Cref{assmp:densitysmoothness,assmp:bandwidth,assmp:kernel} hold. 
Then, for $z =0,1$, we have
\[
\sup_{(w,x) \in \mathscr{S}_{WX}}\ \left| \hat{G} (w,x; z) - \tilde{G} (w,x; z) \right| = o_p(n^{-1/2}).
\] 
\end{lemma}
\noindent
By \Cref{lem:approximation}, it suffices to establish the  limiting distribution of $\tilde G(w, x; 1)-\tilde G(w, x; 0)$ for the asymptotic properties of our test statistics.

Note that 
\begin{multline*}
\tilde G(w,x,z) = \frac{1}{n} \sumi  [\Indi^{*}_{X_iZ_i}(x,z) \hat{q}(X_i,z') (W_i-\hat{W}_i)  \Indi( W_i \le w)] \\
+\frac{1}{n} \sumi [\Indi^{*}_{X_iZ_i}(x,z) \hat{q}(X_i,z')  (w-W_i)  \Indi (W_i\leq w)] 
\equiv \mathbb U_{1}(w,x;z)+\mathbb U_{2}(w,x;z).
\end{multline*}where $z'=1-z$.
Moreover, we have
\[
\mathbb U_1(w,x;z) = \frac{1}{n} \sumi \{ \Indi^{*}_{X_iZ_i}(x,z)  q(X_i,z') (1-D_i) \Indi (W_i\leq w) [{\delta}(X_i)- \hat{\delta}(X_i)] \} + o_p(n^{-\frac{1}{2}})
\]
provided that $\sup_{x\in\mathscr S_X}\left|\left[\hat q(x,z)- q(x,z)\right]\times \left[\hat \delta(x) -\delta(x)\right]\right| = o_p(n^{-\frac{1}{2}})$ holds.
Therefore, the leading terms in $\mathbb {U}_1(w,x;z)$ and $\mathbb{U}_2(w,x;z)$  are essentially two $\mathcal{U}$--processes indexed by $w$ and $x$. 
Following \citet[Theorem 5]{nolan1988functional} and \citet[Theorem 3.1]{powell1989semiparametric}, $\sqrt n[\tilde G(\cdot, \cdot; z)-G(\cdot, \cdot; z)]$ converges in distribution to a mean--zero Gaussian process with bounded and nonzero covariance kernel.

We denote $F_{WD|XZ}(w,d|X,z) \equiv  F_{W|DXZ}(w|d,X,z) \times \Pr (D=d|X,Z=z)$. Let 
\begin{align*}
&\kappa^c(w,x)= -\frac{F_{WD|XZ}(w,0|x,1)-F_{WD|XZ}(w, 0|x,0)}{p(X,1)-p(x,0)};\\
&\psi^c_{wx} =  \Big\{\lambda(w-W) - \mathbb{E}[\lambda(w-W)|X]\Big\} \left[ \frac{\mathbb{1}^{*}_{XZ}(x,0)}{q(X,0)} - \frac{ \mathbb{1}^{*}_{XZ}(x,1)} { q(X,1)}  \right ]  q(X,0) q(X,1);\\
&\phi^c_{wx} =\kappa^c(w, X) \times [W-\mathbb E(W|X)] \times \left[ \frac{\mathbb{1}^{*}_{XZ}(x,0)}{q(X,0)} - \frac{ \mathbb{1}^{*}_{XZ}(x,1)} { q(X,1)}  \right ]q(X,0) q(X,1).
\end{align*} 
By definition, $\psi^c_{wx}$ and $\phi^c_{wx}$ and are random processes indexed by $(w,x)$. Moreover, we have $\mathbb E (\psi^c_{wx}|X,Z)=\mathbb E (\phi^c_{wx}|X,Z)=0$ under $\mathbb H^e_0$.

To establish the weak convergence, we make the following assumptions.
\begin{assumption}
\label{assmp:continuity}
$f_{X|Z}(x|z)$, $\delta(x)$, $p(x,z)$ and $\mathbb E (Y|X=x,Z=z)$ are continuously differentiable in $x$ for $z=0,1$. 
\end{assumption}

\begin{assumption}
\label{assmp:bandkernel}
Let $nh^3_q\rightarrow \infty$ and $nh^3\rightarrow \infty$ as $n\rightarrow \infty$. 
Moreover, The support of $K$ (resp. $K_q$) is a convex (possibly unbounded) subset of $\mathbb R$ with nonempty interior, with the origin as an interior point. 
$K(\cdot)$ (resp. $K_q(\cdot)$) is a bounded differentiable function such that $\int K(u)=1$, $\int uK(u)=0$, and $K(u)=K(-u)$ holds for all $u$ in the support. 
\end{assumption}

\begin{assumption}
\label{assmp:bias}
$\sup_{x\in\mathscr S_X}\left|\mathbb E [\hat \delta(x)]-\delta(x)\right|=o_p(n^{-\frac{1}{2}})$ and $\sup_{xz\in\mathscr S_{XZ}}\left|\mathbb E[ \hat q(x,z)]-q(x,z)\right|=o_p(n^{-\frac{1}{2}})$.
\end{assumption}
\noindent
\Cref{assmp:densitysmoothness} is a smoothness condition that can be further relaxed by the Lipschitz condition. \Cref{assmp:bandkernel} is standard in the kernel regression literature. In particular, the first part strengths  the conditions for bandwidth choice in \Cref{assmp:bandwidth}.
\Cref{assmp:bias} strengths \Cref{assmp:kernel} by requiring the bias term in the first--stage nonparametric estimation to be smaller than $o_p(n^{-1/2})$, which can be established under high order kernels \citep[see e.g.][]{powell1989semiparametric}.

The limiting distribution of our test statistic is summarized in the following theorem.
\begin{theorem}
\label{thm:test_contX}
Suppose the assumptions in \Cref{lem:approximation} and in addition \Cref{assmp:continuity,assmp:bandkernel,assmp:bias} hold. 
Then, under $\mathbb{H}_0^e$, 
\begin{align*}
& \hat{\mathcal{T}}^c_n\overset{d}{\rightarrow }\sup_{w\in\mathbb R; \ x\in\mathscr S_X} |\mathcal Z^c(w,x)|
\end{align*} 
where   $\mathcal Z^c(w,x)$  is a mean--zero Gaussian process with the following covariance kernel
\[
\text{Cov} \left[ \mathcal{Z}^c(w,x), \mathcal{Z}^c(w',x') \right] = \Exp \left[ (\psi^{c}_{wx} + \phi^{c}_{wx}) (\psi^{c}_{w'x'} + \phi^{c}_{w'x'})\right], \ \forall w,w'\in\mathbb R.
\]
Moreover, under $\mathbb H_1^e$, we have  
\[
n^{-\frac{1}{2}}\hat{\mathcal{T}}^c_n  \overset{p}{\rightarrow} \ \sup_{w \in \mathbb{R};  \ x\in\mathscr S_X}\  |G(w,x;0)-G(w,x;1)|.
\]
\end{theorem}
\noindent
It is worth pointing out that our test is one-sided against any alternative to $\mathbb H_0$: reject $\mathbb H_0$ significance level $\alpha$ if and only if $\hat{\mathcal T}^c_n\geq c_{\alpha}$. 

Similar to the discrete--covariates case, we apply the multiplier bootstrap method to approximate the entire process and therefore to approximate critical values. 
The estimates of the pointwise influence function can be estimated by 
\begin{align*}
&\hat{\psi}^c_{wx} = \Big\{\lambda(w-\hat{W}) - \hat{\Exp}[\lambda(w-\hat{W})|X]\Big\} \left[ \frac{\1^{*}_{XZ}(x,0)}{\hat{q}(X,0)} - \frac{ \mathbb{1}^{*}_{XZ}(x,1)} { \hat{q}(X,1)}  \right ] \hat{q}(X,0)\hat{q}(X,1), \\
&\hat \phi^c_{wx} = -\hat{\kappa}^c(w, X) \times \left[\hat W-\hat{\mathbb E} (W|X)\right] \times \left[ \frac{\mathbb{1}^{*}_{XZ}(x,0)}{\hat{q}(X,0)} - \frac{ \mathbb{1}^{*}_{XZ}(x,1)} { \hat{q}(X,1)}  \right ]\hat{q}(X,0)\hat{q}(X,1),
\end{align*} 
where 
\[
\hat{\Exp}[\lambda(w-\hat{W})|X] = \frac{\sum_{j=1}^n \lambda(w-\hat{W}_j)K( \frac{X_j - X}{h})}{\sum_{j=1}^n K( \frac{X_j - X}{h} )} \text{ and }\hat{\mathbb E} (W|X)=\frac{\sum_{j=1}^n \hat{W}_jK( \frac{X_j - X}{h})}{\sum_{j=1}^n K( \frac{X_j - X}{h})}
\]
are estimators of $\mathbb E [\lambda(w-\hat{W})|X]$ and $\mathbb E (W|X)$, respectively, and 
\[
\hat \kappa^c(w,x)=-\frac{\hat F_{WD|XZ}(w,0|x,1)-\hat F_{WD|XZ}(w, 0|x,0)}{\hat p(X,1)-\hat p(x,0)}
\] 
in which
\[
\hat p(X,z)=\frac{\sum_{j=1}^n
D_jK( \frac{X_j - X}{h})\mathbb 1(Z_j=z)}{\sum_{j=1}^n K( \frac{X_j - X}{h})\mathbb 1(Z_j=z)},\ z=0,1
\]
and
\[
\hat F_{WD|XZ}(w,0|x,z)=\frac{\sum_{j=1}^n \mathbb 1 (\hat W_j\leq w;D_j=0)K(\frac{X_j-X}{h})\mathbb 1 (Z_j=z)}{\sum_{j=1}^n K(\frac{X_j-X}{h})\mathbb 1 (Z_j=z)}.
\]

 Note that we can simulate the limiting process by the following result: 
\[
\frac{1}{\sqrt n}\sum_{i=1}^n U_i \times (\hat{\psi}^c_{wx,i} + \hat{\phi}^c_{wx,i})\Rightarrow \mathcal Z^c(w,x).
\]
The test can be constructed similar to the discrete case
and we omit the details for brevity.
%

\section{Monte Carlo Simulations}
\label{sec:simulations}
We investigate the finite sample performance of our tests using Monte Carlo methods.  
First, we examine both size and power by using some simple data generating processes (DGPs). 
Two DGPs are considered. 
\begin{align*}
&\text{DGP 1}: \ Y = D X  + X\epsilon; \ \  D = \Indi(Z - \eta > 0);\\
&\text{DGP 2}: \ Y = D X  + (1+\gamma D)  X \epsilon; \ \ D = \Indi(Z - \eta > 0),
\end{align*} 
where $X$ is uniformly distributed on $\{1,2,3,4,5\}$,  $\epsilon \in \realine$ is uniformly distributed on $(0,1)$, and  $\eta = \rho \epsilon + \sqrt{1-\rho^2}  u$, where  $u$ is uniformly distributed on $(0,1)$ and independent of $\epsilon$, and $\rho=0.5, 0.7$, or $0.9$, respectively. 
The value of $\rho$ represents  level of the endogeneity. 
The instrumental variable $Z \in \{0,1\}$, independent of $X$,  follows Bernoulli distribution with the $\Pr(Z = 1)=0.3$, $0.5$, or $0.7$, respectively. 
Moreover, $\gamma =0.1,0.3$ and $0.5$ represents the ``degree'' of nonseparability in DGP 2. By design, $\mathbb{H}_0$ holds in DGP 1, but not in DGP 2.

We choose the sample size $n=500$, $1000$, and $2000$, and the rejection rate is approximated by $1000$ repetitions. 
To simulate the stochastic processes in the limiting distribution for deriving the critical values, we follow the the method of multiplier bootstrap stated in \Cref{sec:condtional_indep} with 1000 bootstrap repetition.
Moreover, we use $100$ grids on the support of $(\min(\hat{W}),\max(\hat{W}))$ for the suprema  of simulated stochastic processes.  

\Cref{table:dsize} reports size performances under DGP 1 with different values of $\rho$, $\alpha$ and $n$. In all cases, our tests have reasonable size.  \Cref{table:dpower_gamma_1,table:dpower_gamma_3,table:dpower_gamma_5} report power performances under DGP 2 with different values of $\gamma$. In particular,  the rejection rates increase rapidly with the sample size, which verifies the consistency of our test. Moreover, the rejection rate increases with  $\gamma$: Less separable of the error term, more likely to be rejected. 

Next, we investigate the case with continuously distributed covariates $X$. DGP 3 and DGP 4 are the same as DGP 1 and DGP 2, respectively, except that $X$ is a continuous random variable uniformly distributed on $[0,1]$. We choose sample size $n=1000$, $2000$, and $4000$, and the rejection rate is based on $1000$ repetitions. 
For each repetition, the $p$-value is approximated by $500$ simulations. To compute  the suprema of the simulated stochastic processes, we use $100$ grids on the support of $(\min(\hat{W}),\max(\hat{W}))$ and $100$ grids on $[0,1]$ for $w$ and $x$ respectively.
We choose $\gamma= 0.5$ and  $0.7$.

In this continuous covariates case, the performance of the proposed test behaves similarly to  the discrete covariates case. For simplicity, we only present the results for $\Pr(Z=1)=0.5$ and $\rho=0.7$. The results for other settings exhibit similar patterns. \Cref{table:csize} reports empirical levels at various nominal levels.
The level of our test is fairly well behaved and it converges to the nominal level as the sample size increases.
\Cref{table:cpower} presents the empirical power of our test. Clearly, our test is consistent. 

\section{Empirical Applications}
\label{sec:empirical}
\subsection{The Effect of Job Training Program on Earnings}
In this section we apply our testing approach to the effect of job training program on earnings.  The \textit{National Job Training Partnership Act} (JTPA), commissioned by the Department of Labor, began funding training from 1983 to late 1990's to increase employment and earnings for participants.
The major component of JTPA aims to support training for the economically disadvantaged.

Our sample consists of 11,204 observations from the JTPA, a survey  dataset from over $20,000$ adults and out-of-school youths who applied for JTPA in $16$ local areas across the country between 1987 and 1989.\footnote{JTPA services are provided at 649 sites, which might not be randomly chosen. For a given site, the applicants were randomly selected for the JTPA dataset. }
Each participant were assigned randomly to either a program group or a control group (1 out of 3 on average).  Program group are eligible to participate JTPA services, including classroom training, on-the-job training or job search assistance, and other services, while members of control group were not eligible for JTPA services for 18 months.

Following the literature, we use the program eligibility as an instrumental variable for the endogenous individual's participation decision. The outcome variable is individual earnings, measured by the sum of earnings in the 30-month period following the offer. The effects of JTPA training programs on earnings has also been studied by  \cite{abadie2002instrumental} under a general framework allowing for unobserved heterogeneous treatment effects. \footnote{The data is publicly available at \href{http://upjohn.org/services/resources/employment-research-data-center/national-jtpa-study}{http://upjohn.org/services/resources/employment-research-data-center/national-jtpa-study}.}

The observed covariates include a set of dummies for races, for high-school graduates, and for marriage,  whether the applicant worked at least 12 weeks in the 12 months preceding random assignment, and also 5 age-group dummies (22-24, 25-29, 30-35, 36-44, and 45-54), among others.
See \Cref{table:jtpa} for descriptive statistics. For simplicity, we group all applicants into 3 age categories (22-29, 30-35, and 36 and above), and pool all non-White applicants as minority applicants.

To implement, we use the Gaussian kernel with Silverman bandwidth selection.
For the critical value, we use $10,000$ bootstrapped simulations and search for the suprema from $5,000$ grids.
The $p$-value of our test is $0.1204$.
Therefore, $\mathbb H_0$, i.e. there is no unobserved heterogenous treatment effects, cannot be rejected at a $10\%$ significance level. 
It is worth noting that our results are robust to the number of simulations and the number of grids.

\subsection{The Impact of Fertility on Family Income}
The second empirical illustration example is to explore the impact of children on parents' labor supply and income.
 The endogeneity issue arises due to the  fertility variable.
\cite{rosenzweig1980testing} suggest the usage of the twin births as an instrumental variable.
The ``twin-strategy'' IV has been widely used in the literature.
See eg. \cite{angrist1998children} and \cite{vere2011fertility}. The impact of fertility on family income has also been studied by \cite{frolich2013unconditional} under a general framework allowing for unobserved heterogeneous treatment effects.

Our sample  come from 1990 and 2000 censuses, consisting of 602,767 and 573,437 observations, respectively. Similar to \cite{frolich2013unconditional}, we use the 1\% and 5\% Public Use Microdata Sample (PUMS) from the 1990 and 2000 censuses.\footnote{The data is publicly available at \href{https://www.census.gov/main/www/pums.html}{https://www.census.gov/main/www/pums.html}.} Moreover, our sample is restricted to 21--35 years old married mothers with at least one child. 
The outcome variable of interest  is the family's annual labor income.\footnote{It includes wages, salary, armed forces pay, commissions, tips, piece-rate payments, cash bonuses earned before deductions were made for taxes, bonds, pensions, union dues, etc. See \cite{frolich2013unconditional} for more details.} The treatment variable  is the dummy for a mother has two or more children. The instrument variable is the dummy for the first birth is a twin.  The covariates includes mother's and father's age, race, educational level, and working status. \Cref{table:fertility} provides descriptive statistics. Some covariates, e.g., age, years in education, and working hours per week, are treated as continuous variables.

For the critical value,  we use $5000$ bootstrapped simulations and search for the suprema from $1000$ grids. 
The $p$-values of our tests are $0.0031$ and $0.0004$ for the 1990 and 2000 censused, respectively.
These results suggest that the null hypothesis, i.e., homogeneous treatment effects, should be rejected at all usual significance levels.

%
%

\newpage
\begin{appendices}
\section*{Appendix}
\renewcommand{\theequation}{\Alph{section}.\arabic{equation}}
\setcounter{equation}{0}

\section{Proofs of Lemmas and Theorems}
\label{sec:appendix}
\subsection{Proof of \Cref{prop:equivalence bw hetero and separability}}
\begin{proof}
For the ``only if'' part, under $\mathbb{H}_0$, we have
\[
g(1,x,\epsilon) - h(0,x,\epsilon) = m(1,x) - m(0,x) \equiv \delta(x), \ \ \forall x \in \mathscr{S}_X.
\]

For the ``if part'', \eqref{counterfactual} implies 
\[
g(d,x,\epsilon) =  d \times [ g(1,x,\epsilon)  - g(0,x,\epsilon) ] + g(0,x,\epsilon) =  d \times \delta(x) + g(0,x,\epsilon).
\]
Therefore, $\mathbb{H}_0$ holds in the sense $m(d,x) = d \times \delta(x)$ and $\nu(x,\epsilon) = g(0,x,\epsilon)$.
\end{proof}


\subsection{Proof of \Cref{thm:equivalence bw separability and conditinoal Indep}}
\begin{proof}
Given \Cref{prop:equivalence bw hetero and separability}, it suffices to show the if part. 
Suppose $W \indep Z \cond X$ holds. 
Recall that $W = Y + (1-D) \times \delta(X)$. 
It follows that
\begin{multline*}
\Pr(Y\le y, D=1 | X, Z = 1 ) + \mathbb{P}(Y + \delta(X) \le y, D=0 | X, Z = 1 ) \\
=\Pr(Y\le y, D=1 | X, Z = 0 ) + \mathbb{P}(Y + \delta(X) \le y, D=0 | X, Z = 0 ), \ \forall y \in \mathbb{R}.
\end{multline*} 
Equivalently, we have
\begin{multline} 
\label{eqn:a1}
\Pr(Y \le y, D=1 | X, Z = 1 ) - \mathbb{P}(Y \le t, D=1 | X, Z = 0 )\\
=\Pr(Y \le y-\delta(X), D=0 | X, Z = 1 ) - \mathbb{P}(Y \le y-\delta(X), D=0 | X, Z = 0 ).
\end{multline} 

Let $V \equiv \nu(X,\epsilon)$ and define 
\begin{align*}
\Delta_0(\tau,x) &\equiv \Pr(V \le \tau, D=0 | X = x, Z = 1 ) - \Pr(V \le \tau, D=0 | X = x, Z = 0 ); \\
\Delta_1(\tau,x) &\equiv \Pr(V \le \tau, D=1 | X = x, Z = 0 ) - \Pr(V \le \tau, D=1 | X = x, Z = 1 ).
\end{align*}
By \Cref{assmp:relevant IV,assmp:monotone selection}, we have
\[
\Delta_0(\tau,x) = \Pr(V \le \tau, \eta \in \mathcal{C}_x | X = x) = \Delta_1(\tau,x)
\] 
which is strictly monotone in $\tau \in \mathscr S_{V|X=x,\ \eta \in \mathcal{C}_x}$ and, by \Cref{assmp:single index error,assmp:sufficiently large compliers}, 
\[
\mathscr S_{V|X=x,\ \eta \in \mathcal{C}_x}= \mathscr S_{V|X=x}.
\]
Therefore, we have 
\begin{align*}
&\Pr(Y \le y, D=1 | X = x, Z = 0 ) - \mathbb{P}(Y \le y, D=1 | X = x, Z = 1 )\\
&\quad=\Delta_1(\tilde{g}^{-1}(1,x,y), x)=\Delta_0(\tilde{g}^{-1}(1,x,y), x)\\
&\quad=\Pr(Y \le \tilde{g}(0,x,\tilde{g}^{-1}(1,x,y)), D=0 | X = x, Z = 1 ) \\
&\qquad- \mathbb{P}(Y \le \tilde{g}(0,x,\tilde{g}^{-1}(1,x,y)), D=0 | X = x, Z = 0 ),
\end{align*} 
where $\tilde{g}^{-1}(1,x,\cdot )$ is the inverse function of $\tilde{g}(1,x, \cdot)$. 
Note that both sides are strictly monotone in $y\in \mathscr S_{\tilde{g}(1,X,V)|X=x}$ since $\Delta_d(\cdot,x)$ is strictly monotone on $\mathscr S_{V|X=x}$.

Combining the above result with \eqref{eqn:a1}, we have 
\[
\tilde{g}(0,x,\tilde{g}^{-1}(1,x,y))= y-\delta(x), \ \ \ \forall x\in\mathscr S_X, \ y\in  \mathscr S_{\tilde{g}(1,x,V)|X=x}.
\] 
Let $y= \tilde{g}(1,x,\tau)$ for $\tau \in \mathscr{S}_{V|X=x}$. 
It follows that 
\[
\tilde{g}(0,x,\tau) = \tilde{g}(1,x,\tau) -\delta(x),
\]
which gives us the result by \Cref{prop:equivalence bw hetero and separability}.
\end{proof}

\subsection{Proof of \Cref{thm:heterogeneity and heteroscadesticity}}
\begin{proof}
For notational simplicity, let $W^* = W - \mathbb{E}(W|X)$. For the ``only if'' part, under $\mathbb{H}^{\dag}_{0}$, since $W^* = h(1,X,\epsilon) - \Exp[h(1,X,\epsilon)|X]$ is a function of $\epsilon$, we have $W^* \indep (X,Z)$.

For the ``if'' part, first note that $W^* \indep Z | X$ is equivalent to $W \indep Z | X$. 
Then, by \Cref{thm:equivalence bw separability and conditinoal Indep}, $g(D,X,\epsilon) = m(D,X) + \nu(X,\epsilon)$ and
\[
W^* = m(1,X) + \nu(X, \epsilon) - \Exp[m(1,X) + \nu(X, \epsilon)|X] = \nu(X, \epsilon) - \Exp[\nu(X, \epsilon)|X] = \overline{\nu}(X,\epsilon)
\]
is a function of $X$ and $\epsilon$. 
Note that $\overline{\nu}(X,\epsilon)$ is strictly increasing in $\epsilon$. 
We now show that there exists a measurable function $\nu_1: \mathbb{R} \mapsto \mathbb{R}$ such that $\nu(x,e) = \nu_1(e)$ almost everywhere.
To see this, first note that $\epsilon \indep X$ implies $\Pr(\epsilon > e | X=x) = \Pr(\epsilon > e)$ for almost all $x$.
We also have
\begin{align*}
\Pr(\nu(x,\epsilon) > \tilde{e} | X=x) = \Pr(\epsilon > \nu^{-1}(x,\tilde{e})|X=x) = \Pr(\epsilon >  \nu^{-1}(x,\tilde{e}))
\end{align*}
where $\nu^{-1}$ is the inverse function of $\nu$ with respect to its second argument.
Since the c.d.f. of $\epsilon$ is strictly increasing and $\Pr(\nu(x,\epsilon) > \tilde{e}|X=x)$ does not depend $x$, $\nu^{-1}(x,\tilde{e})$ must be constant in $x$ for almost all $(x,\tilde{e})$. 
Thus, $\nu(x,e)$ is constant in $x$ almost everywhere and this completes the proof.

\end{proof}

\subsection{Proof of \Cref{thm:test_noX}}
\begin{proof}
For discrete $X$, we have $\hat\delta^*(\cdot)=\delta^*(\cdot)+O_p(n^{-1/2})$ by the central limit theorem. 
The $O_p(n^{-1/2})$ term is uniformly over the finite support $\mathscr S_X$. Moreover, let $\mathbb{1}_{WXZ}(w,x,z;\delta) = \mathbb 1(W \leq w) \cdot \mathbb 1_{XZ}(x,z)$.
By definition, we have
\[
F_{W|XZ}(w|x,z_\ell)=\frac{\Exp [\mathbb{1}_{WXZ} (w,x,z_\ell;\delta)]}{\Exp [\mathbb 1_{XZ}(x,z)]}\ \text{ and }\ \hat{F}_{W|XZ}(w|x,z) = \frac{\Exp_n [\mathbb{1}_{WXZ}(w,x,z;\hat{\delta})]}{\mathbb E_n [\mathbb 1_{XZ}(x,z)]}.
\]
Note that 
\begin{align*}
&\Exp_n [\mathbb{1}_{WXZ} (\cdot, x, z; \hat{\delta})] 
= \Exp_n [ \mathbb{1}_{WXZ} (\cdot, x, z; \delta)] - \Exp [\mathbb{1}_{WXZ} (\cdot, x, z; \delta)]\\
&\qquad+ \Big\{ \Exp_n [\mathbb{1}_{WXZ} (\cdot, x, z; \hat{\delta})] - \Exp [\mathbb{1}_{WXZ}(\cdot, x, z; \hat{\delta})] - \Exp_n [ \mathbb{1}_{WXZ} (\cdot, x, z; \delta)] + \Exp [\mathbb{1}_{WXZ} (\cdot, x, z; \delta) ] \Big\} \\ 
&\qquad+\Exp [ \mathbb{1}_{WXZ} (\cdot, x, z; \hat{\delta})].
\end{align*}
Since $\hat{\delta}$ is a consistent estimator of $\delta$, by the empirical process theory \citep[see e.g.][]{van2007empirical}, we have
\[
\mathbb E_n [ \mathbb{1}_{WXZ} (\cdot, x, z; \hat{\delta})] - \mathbb E [ \mathbb{1}_{WXZ} (\cdot, x, z; \hat{\delta})] - \mathbb E_n [ \mathbb{1}_{WXZ} (\cdot, x, z; \delta)] + \mathbb E [\mathbb{1}_{WXZ} (\cdot, x, z ;\delta)] =o_p(n^{-1/2}).
\]
Moreover, by Taylor expansion, 
\[
\sqrt n\ \mathbb E [g(\cdot, x, z; \hat{\delta})] = \sqrt{n}\ \Pr(W \leq w, X = x, Z = z) + \frac{\partial \mathbb E [\mathbb{1}_{WXZ}(w, x, z; {\delta})]}{\partial {\delta}} \times \sqrt{n} (\hat{\delta} - \delta)+o_p(1)
\] 
where 
\[
\frac{\partial \mathbb{E} [ \mathbb{1}_{WXZ} (w, x, z; \delta)]}{\partial {\delta}(x')}=0\ \text{for all } x'\neq x
\] 
and 
\begin{align*}
\frac{\partial \mathbb{E} [\mathbb{1}_{WXZ} (w, x, z;\delta)]}{\partial {\delta}(x)} &= -f_{Y|DXZ}(w-\delta(x)|0,x,z)\times \Pr (D=0,X=x,Z=z) \\
&= -f_{W|DXZ}(w|0,x,z)\times \Pr (D=0,X=x,Z=z).
\end{align*} 

It follows that
\begin{multline*}
\sqrt n \Exp_n [\mathbb{1}_{WXZ} (\cdot, x, z; \hat{\delta})]
= \sqrt n \left\{\Exp_n [ \mathbb{1}_{WXZ} (\cdot, x, z; \delta)] - \Exp [\mathbb{1}_{WXZ} (\cdot, x, z; \delta)]\right\}\\
+ \sqrt n\  \Pr(W \leq w, X = x, Z = z) + \frac{\partial \mathbb E [\mathbb{1}_{WXZ}(w, x, z; {\delta})]}{\partial {\delta}(x)} \times \sqrt{n} (\hat{\delta}(x) - \delta(x))+o_p(1).
\end{multline*}
Moreover, by the central limit theorem, we have
\[
\mathbb E_n [\mathbb 1_{XZ}(x,z_\ell)]=\mathbb P (X=x,Z=z) + O_p(n^{-1/2}).
\]

Thus,
\begin{align*}
&\sqrt{n} \left[\hat F_{W|XZ}(w|x,1)-\hat F_{W|XZ}(w |x,0)\right]\\
&= \frac{\sqrt{n} \big\{\mathbb E_n [\mathbb{1}_{WXZ}(w,x,1;\delta)] - \mathbb E [\mathbb{1}_{WXZ}(w,x,1;\delta)] \big\} + \frac{\partial \mathbb{E} [\mathbb{1}_{WXZ}(w, x, 1;\delta)]}{\partial {\delta}(x)} \cdot \sqrt{n} [\hat{\delta}(x)-\delta(x)]}{\mathbb E_n \mathbb 1 _{XZ}(x,1)} \\
&\quad-\frac{\sqrt{n} \big\{\mathbb  E_n [\mathbb{1}_{WXZ}(w,x,0;\delta)]  - \mathbb E [\mathbb{1}_{WXZ}(w,x,0;\delta)] \big\} + \frac{\partial \mathbb{E} [\mathbb{1}_{WXZ} (w, x, 0;\delta)]}{\partial {\delta}(x)} \cdot \sqrt{n} [\hat{\delta}(x)-\delta(x)]}{\mathbb E_n \mathbb 1 _{XZ}(x,0)} \\
&\quad+\frac{\sqrt{n} \Pr(W \leq w, X = x, Z = 1)}{\mathbb E_n \mathbb 1_{XZ}(x,1)} -\frac{\sqrt{n} \Pr(W \leq w, X = x, Z = 0)}{\mathbb E_n \mathbb 1 _{XZ}(x,0)} +o_p(1)\\
&= \frac{\sqrt{n} \big\{\mathbb E_n [\mathbb{1}_{WXZ} (w,x,1;\delta)] - \mathbb E [\mathbb{1}_{WXZ} (w,x,1;\delta)] \big\} + \frac{\partial \mathbb{E} [\mathbb{1}_{WXZ} (w, x, 1;\delta)]}{\partial {\delta}(x)} \cdot \sqrt{n} [\hat{\delta}(x)-\delta(x)]}{\mathbb P(X=x,Z=1)} \\
&\quad-\frac{\sqrt{n} \big\{\mathbb  E_n [ \mathbb{1}_{WXZ} (w,x,0;\delta)]  - \mathbb{E} [\mathbb{1}_{WXZ}(w,x,0;\delta)] \big\} + \frac{\partial \mathbb{E} [\mathbb{1}_{WXZ} (w, x, 0;\delta)]}{\partial {\delta}(x)} \cdot \sqrt{n} [\hat{\delta}(x)-\delta(x)]}{\mathbb P(X=x,Z=0)} \\
&\quad+\frac{\sqrt{n} \Pr(W \leq w, X = x, Z = 1)}{\mathbb E_n \mathbb 1_{XZ}(x,1)} -\frac{\sqrt{n} \Pr(W \leq w, X = x, Z = 0)}{\mathbb E_n \mathbb 1 _{XZ}(x,0)} +o_p(1).
\end{align*}
Moreover, we apply Taylor expansion and have 
\begin{align*}
&\frac{\sqrt{n} \Pr(W \leq w, X = x, Z = z)}{\mathbb E_n \mathbb 1_{XZ}(x,z)}\\ &=\frac{\sqrt{n} \Pr(W \leq w, X = x, Z = z)}{\mathbb E \mathbb 1_{XZ}(x,z)} - \frac{ \Pr(W \leq w, X = x, Z = z)}{\mathbb E^2 \mathbb 1_{XZ}(x,z)}\times \sqrt{n}\big[\mathbb E_n \mathbb 1 _{XZ}(x,z)-\mathbb E \mathbb 1 _{XZ}(x,z)\big] +o_p(1)\\
&= \sqrt{n}\ F_{W|X}(w|x)-F_{W|X}(w|x)\times \frac{\sqrt n  \big[\mathbb E_n \mathbb 1_{XZ}(x,z) - \mathbb E  \mathbb 1_{XZ}(x,z)\big]}{\Pr(X=x,Z=z)} +o_p(1)
\end{align*}
where the last step comes from $F_{W|XZ}(\cdot|x,z)=F_{W|X}(\cdot|x)$ under the  null hypothesis.  
Note that
\begin{eqnarray*}
F_{W|X}(w|x) \times \Exp \mathbb 1_{XZ}(x,z)  = \Pr(W \le w, X=x, Z=z)=\mathbb E [\mathbb{1}_{WXZ}(w,x,z;\delta)] .
\end{eqnarray*}
Thus, we have
\begin{align*}
&\frac{\sqrt{n} \Pr(W \leq w, X = x, Z = z)}{\mathbb E_n \mathbb 1_{XZ}(x,z)}\\
&\quad=\sqrt{n}\ F_{W|X}(w|x)- \frac{\sqrt n  \left\{F_{W|X}(w|x)\times\mathbb E_n \mathbb 1_{XZ}(x,z) - \mathbb E [\mathbb{1}_{WXZ}(w,x,z;\delta)] \right\}}{\Pr(X=x,Z=z)}.
\end{align*}

Summarizing above steps, we now have
\begin{align*}
&\sqrt{n} \left[\hat F_{W|XZ}(w|x,1) - \hat F_{W|XZ}(w |x,0)\right]\\
&=\sqrt{n} \Exp_n \left\{[\Indi(W \le w) - F_{W|X}(w|x)] \times \left[ \frac{\mathbb 1_{XZ}(x,1)}{ \mathbb P (X=x,Z=1)} -  \frac{\mathbb 1_{XZ}(x,0)}{ \mathbb P (X=x,Z=0)}  \right ] \right\}\\
&\quad+ \left[ f_{WD|XZ}(w,0|x,0)-f_{WD|XZ}(w, 0|x,1)\right]\times \sqrt n [\hat \delta(x)-\delta(x)]+o_p(1).
\end{align*}

To compute the covariance kernel, we first derive $\hat{\delta} - \delta$.
Fix $X=x$. For expositional simplicity, we suppress $x$ in the following notation. Note that 
\[
\hat \delta=\frac{\mathbb A_{n}(1) \mathbb C_n(0)-\mathbb A_{n}(0)\mathbb C_n(1)}{\mathbb B_n(1) \mathbb C_n(0)-\mathbb B_n(0)\mathbb C_n(1)} \ \ \text{ and } \  \delta=\frac{\mathbb A(1) \mathbb C(0)-\mathbb A(0)\mathbb C(1)}{\mathbb B(1) \mathbb C(0)-\mathbb B(0)\mathbb C(1)}, 
\] where $\mathbb A _{n}(z)=\mathbb E_n [Y\cdot \mathbb 1 (X=x,Z=z)]$, $\mathbb B_{n}(z) =\mathbb E_n [D\cdot \mathbb 1 (X=x,Z=z)]$, $\mathbb C_{n}(z) =\mathbb E_n \ \mathbb 1 (X=x,Z=z)$, $\mathbb A (z)=\mathbb E [Y\cdot \mathbb 1 (X=x,Z=z)]$, $\mathbb B (z)=\mathbb E [D\cdot \mathbb 1 (X=x,Z=z)]$, and $ \mathbb C (z)=\mathbb P (X=x,Z=z)$. Therefore, 
\begin{multline*}
\hat \delta-\delta=\frac{\mathbb A_{n}(1) \mathbb C_n(0)-\mathbb A_{n}(0)\mathbb C_n(1)-[\mathbb A(1) \mathbb C(0)-\mathbb A(0)\mathbb C(1)]}{\mathbb B_n(1) \mathbb C_n(0)-\mathbb B_n(0)\mathbb C_n(1)}\\
+\left\{ \frac{\mathbb A(1) \mathbb C(0)-\mathbb A(0)\mathbb C(1)}{\mathbb B_n(1) \mathbb C_n(0)-\mathbb B_n(0)\mathbb C_n(1)}-\frac{\mathbb A(1) \mathbb C(0)-\mathbb A(0)\mathbb C(1)}{\mathbb B(1) \mathbb C(0)-\mathbb B(0)\mathbb C(1)}\right\}
\equiv \mathbb I+ \mathbb {II}.
\end{multline*}
Let  $p_c(x)=\mathbb P(X=x,Z=1) - \Pr(X=x,Z=0)$, which is strictly positive under \Cref{assmp:relevant IV}.
 
Note that
\begin{align*}
\mathbb I &= \frac{\left[\mathbb A_{n}(1)-\mathbb A(1)\right]\cdot  \mathbb C_n(0) + \mathbb A(1)\cdot \left[\mathbb C_n(0)-\mathbb C(0)\right]}{\mathbb B_n(1) \mathbb C_n(0)-\mathbb B_n(0)\mathbb C_n(1)}\\
&\quad- \frac{\left[\mathbb A_{n}(0)-\mathbb A(0)\right]\cdot \mathbb C_n(1)+\mathbb A(0)\cdot \left[ \mathbb C_n(1)-\mathbb C(1)\right]}{\mathbb B_n(1) \mathbb C_n(0)-\mathbb B_n(0)\mathbb C_n(1)}\\
&=\frac{\left[\mathbb A_{n}(1)-\mathbb A(1)\right]\cdot  \mathbb C(0) + \mathbb A(1)\cdot \left[\mathbb C_n(0)-\mathbb C(0)\right]}{\mathbb B(1) \mathbb C(0)-\mathbb B(0)\mathbb C(1)}\\
&\quad- \frac{\left[\mathbb A_{n}(0)-\mathbb A(0)\right]\cdot \mathbb C(1)+\mathbb A(0)\cdot \left[ \mathbb C_n(1)-\mathbb C(1)\right]}{\mathbb B(1) \mathbb C(0)-\mathbb B(0)\mathbb C(1)}+o_p(n^{-1/2})
\end{align*} 
where the last step comes from the fact: $\mathbb A_n(z)=\mathbb A(z)+O_p(n^{-1/2})$, $\mathbb B_n(z)=\mathbb B(z)+O_p(n^{-1/2})$ and $\mathbb C_n(z)=\mathbb C(z)+O_p(n^{-1/2})$. 
Therefore,
\begin{align*}
\mathbb I 
&= \frac{\mathbb E_n \left\{\big[Y -\mathbb E (Y|X=x,Z=0)\big]\cdot \mathbb 1_{XZ}(x,1) \right\}\times \Pr(X=x,Z=0)}{\mathbb B(1) \mathbb C(0)-\mathbb B(0)\mathbb C(x,1)}\\
&\quad- \frac{\mathbb E_n \left\{\big[Y -\mathbb E (Y|X=x,Z=1)\big]\cdot \mathbb 1_{XZ}(x,0) \right\}\times \Pr(X=x,Z=1)}{\mathbb B(1) \mathbb C(0)-\mathbb B(0)\mathbb C(1)}\\
&\quad+\frac{2\left[\mathbb A(0)\cdot  \mathbb C(1)-\mathbb A(1)\cdot  \mathbb C(0)\right] }{\mathbb B(1) \mathbb C(0)-\mathbb B(0)\mathbb C(1)}+o_p(n^{-1/2})\\
&= -\frac{1}{p_c(x)}\times \mathbb E_n \left\{\big[Y -\mathbb E (Y|X=x,Z=0)\big]\times\frac{\mathbb 1_{XZ}(x,1)}{\mathbb P(X=x,Z=1)}\right\}\\
&\quad+ \frac{1}{p_c(x)}\times \mathbb E_n \left\{\big[Y -\mathbb E (Y|X=x,Z=1)\big]\times \frac{\mathbb 1_{XZ}(x,0)}{\mathbb P (X=x,Z=0)} \right\} \\
&\quad- 2\delta(x) + o_p(n^{-1/2})
\end{align*} 
where the last step comes from
\[
\mathbb B(1) \mathbb C(0)-\mathbb B(0)\mathbb C(1)=-p_c(x)\times \mathbb P (X=x,Z=1)\times \mathbb P (X=x,Z=0).
\]
Similarly, by Taylor expansion, 
\begin{align*}
\mathbb {II}&= -\delta(x)\times \frac{\left[\mathbb B_n(1)-\mathbb B(1)\right]\cdot \mathbb C(0)+\mathbb B(1) \cdot \left[\mathbb C_n(0)-\mathbb C(0)\right]}{\mathbb B(1) \mathbb C(0)-\mathbb B(0)\mathbb C(1)}\\
&\quad+\delta(x)\times \frac{\left[\mathbb B_n(0)-\mathbb B(0)\right]\cdot \mathbb C(1)+\mathbb B(0) \cdot \left[\mathbb C_n(1)-\mathbb C(1)\right]}{\mathbb B(1) \mathbb C(0)-\mathbb B(0)\mathbb C(1)}+o_p(n^{-1/2})\\
&=\frac{1}{p_c(x)}\times \mathbb E_n \left\{\big[D -p(x,0)\big]\times  \delta(x)\times \frac{\mathbb 1_{XZ}(x,1)}{\mathbb P(X=x,Z=1)}\right\} \\
&\quad-\frac{1}{p_c(x)}\times \mathbb E_n \left\{\big[D -p(x,1)\big]\times  \delta(x)\times \frac{\mathbb 1_{XZ}(x,0)}{\mathbb P(X=x,Z=0)}\right\} \\
&\quad+ 2  \delta(x) +o_p(n^{-1/2}).
\end{align*}
Note that under the null hypothesis $\mathbb H_0$:
\begin{multline*}
\mathbb E (Y|X=x,Z=1)-\mathbb E (D|X=x,Z=1)\times \delta (x)=\mathbb E [W-\delta(X)|X=x,Z=1]\\
=\mathbb E [W-\delta(X)|X=x,Z=0]=\mathbb E (Y|X=x,Z=0)-\mathbb E (D|X=x,Z=0)\times \delta (x).
\end{multline*}
Moreover, we have
\[
Y-D\cdot \delta(X) -\mathbb E (Y|X,Z)+ \mathbb E (D|X,Z)\cdot \delta(X)= W-\mathbb E (W|X,Z)=W-\mathbb E (W|X).
\]

Thus, we have 
\begin{multline} \label{delta_influence}
\sqrt n\ [\hat \delta(x)-\delta(x)] \\
= -\frac{1}{p_c(x)}\times \sqrt n \ \mathbb E_n \left\{\overline W\times \left[\frac{\mathbb 1 _{XZ}(x,1)}{\mathbb P(X=x,Z=1)}-\frac{\mathbb 1 _{XZ}(x,0)}{\mathbb P(X=x,Z=0)}\right]\right\}+o_p(1).
\end{multline}

It follows that 
\begin{eqnarray*}
 \sqrt n \left[\hat F_{W|XZ}(w|x,1)-\hat F_{W|XZ}(w |x,0)\right] = \frac{1}{\sqrt{n}} \sumi \left( \psi_{wx} + \phi_{wx} \right).
\end{eqnarray*}
where $\psi_{wx}$ and $\phi_{wx}$ are defined by \eqref{eqn:dpsi} and \eqref{eqn:dphi}.

To conclude, the empirical process $\sqrt n \left[\hat{F}_{W|XZ}(\cdot | x, 1) - \hat {F}_{W|XZ}(\cdot |x, 0)\right]$ converges to a zero--mean Gaussian process $\mathcal{Z}(\cdot,x)$ with the given covariance kernel. Moreover, following e.g., \cite{kim1990cube}, we have $\hat{\mathcal T}_n  \overset{d}{\rightarrow} \ \sup_{w \in \mathbb{R};  \ x\in\mathscr S_X}\  |\mathcal{Z}(w,x)|$.
\end{proof}


\subsection{Proof of \Cref{lem:approximation}}
\begin{proof}
Fix $X=x$ and w.l.o.g., let $z=1$. 
Note that
\begin{align*}
&\hat G(w,x,1)-\tilde G(w,x,1) \\
&=  \mathbb E_n \left\{ \mathbb{1}^{*}_{XZ}(x,1)  \hat{q}(X,0)(w-\hat W) \left[ \mathbb 1(\hat W \leq w) - \mathbb 1( W \leq w)\right] \right\} \\
&=  \mathbb E_n \Big\{ \mathbb{1}^{*}_{XZ}(x,1) \hat{q}(X,0)(w-\hat W) \left[ \mathbb 1(\hat W \leq w) - \mathbb 1( W \leq w) \right] \times \mathbb 1 (|W-w|\leq  n^{-r}) \Big\}\\
&\quad+ \mathbb E_n \Big\{ \mathbb{1}^{*}_{XZ}(x,1)  \hat{q}(X,0)(w-\hat W)  \left[ \mathbb 1(\hat W \leq w)-\mathbb 1( W \leq w)\right] \times  \mathbb 1 (|W-w|>  n^{-r}) \Big\}\\
&\equiv \mathbb T_1+\mathbb T_2
\end{align*}where $r\in(\frac{1}{4}, \iota)$. It suffices to show both $\mathbb T_1$ and $\mathbb T_2$ are $o_p(n^{-\frac{1}{2}})$.

For term $\mathbb T_1$, note that 
\begin{multline*}
\mathbb T_1= \mathbb E_n\Big\{\mathbb{1}^{*}_{XZ}(x,1) \hat{q}(X, 0)  (w- W)  \left[ \mathbb 1(\hat W \leq w)-\mathbb 1( W\leq w) \right] \times  \mathbb 1 (|W-w|\leq  n^{-r})\Big\}\\
+\mathbb E_n \Big\{\mathbb{1}^{*}_{XZ}(x,1) \hat{q}(X,0)( W- \hat W)\left[ \mathbb 1(\hat W \leq w)-\mathbb 1( W \leq w)\right] \times  \mathbb 1 (|W-w|\leq  n^{-r})\Big\}. 
\end{multline*}
Because 
\begin{align*}
&\mathbb E\left|\mathbb{1}^{*}_{XZ}(x,1) \hat{q}(X, 0)  (w- W)  \left[ \mathbb 1(\hat W \leq w)-\mathbb 1( W\leq w) \right] \times  \mathbb 1 (|W-w|\leq  n^{-r})\right| \\
&\quad\leq \mathbb E \left| \hat{q}(X_1,z_{2})\times (w- W)\times  \mathbb 1 (|W-w|\leq  n^{-r})\right|=O(1)\times O(n^{-2r})=o(n^{-\frac{1}{2}}),
\end{align*}
where last step holds because $r>\frac{1}{4}$. Moreover, 
\begin{align*}
&\mathbb E \left|\mathbb{1}^{*}_{XZ}(x,1) \hat{q}(X,0)( W- \hat W)\left[ \mathbb 1(\hat W \leq w)-\mathbb 1( W \leq w)\right] \times  \mathbb 1 (|W-w|\leq  n^{-r})\right|\\
&\quad\leq \mathbb E \left| \hat{q}(X_1,z_{2})(W-\hat W)\times  \mathbb 1 (|W-w|\leq  n^{-r})\right|=O(1)\times O(n^{-\iota})\times O(n^{-r})=o(n^{-\frac{1}{2}}).
\end{align*}
Then, we have $\mathbb T_1 = o_p(n^{-\frac{1}{2}})$. 
 
Next, for term $ \mathbb T_2$, note that
\begin{eqnarray*}
\mathbb E | \mathbb T_2 |&\leq& \frac{\overline K}{h}\times \mathbb E\left[|w-\hat W|\times \mathbb 1 (|\hat W-W|>n^{-r})\right]\\
 &\leq& \frac{\overline K}{h}\times \sqrt {\mathbb E (w-\hat W)^2}\times \sqrt {\mathbb P \left(|\hat W-W|>n^{-r}\right)}\\
 &\leq&\frac{\overline K}{h}\times \sqrt{\mathbb E \hat W^2-2w\cdot \mathbb E (\hat W)+ w^2} \times \sqrt {\mathbb P \left[|\hat \delta(X)-\delta(X)|>n^{-r}\right]},
\end{eqnarray*}
where $\overline K$ is the upper bound of $K(\cdot)$. Because $W$ is a bounded random variable and $w$ belongs to a compact set, then $ \sqrt{\mathbb E \hat W^2-2w\cdot \mathbb E (\hat W)+ w^2}=O(1)$. Moreover, by \Cref{lem:smalldistance1}, $\mathbb E | \mathbb T_2|\leq o(n^{-k})$ for any $k>0$. Hence, $ \mathbb{T}_2 = o_p(n^{-\frac{1}{2}})$.
\end{proof}

\subsection{Proof of \Cref{thm:test_contX}}
\begin{proof}
By \Cref{lem:approximation}, we have 
\[
\hat {\mathcal T}^c_n=\sqrt{n} \left| \tilde G(w, x; 1) - \tilde G(w, x; 0) \right|+o_p(1).
\] Let $\mathbb 1^*_{WXZ}(w,x,z)\equiv \mathbb 1 (W\leq w,X\leq x,Z=z)$. 
Note that 
\[
\tilde G(w,x,z_\ell) = \mathbb U_{1}(w,x;z_\ell)+\mathbb U_{2}(w,x;z_\ell)+o_p(n^{-1/2})
\]
where  $ \mathbb U_{1}(w,x;z_\ell)\equiv  \frac{1}{n}\sum_{i=1}^n  [\Indi^{*}_{W_iX_iZ_i}(w,x,z_{\ell})\times  \hat{q}(X_i,z_{-\ell})\times(W_i-\hat W_i)]$ and  $\mathbb U_{2}(w,x;z_\ell)\equiv \frac{1}{n}\sum_{i=1}^n [\Indi^{*}_{W_iX_iZ_i}(w,x,z_{\ell}) \times  \hat{q}(X_i,z_{-\ell})\times(w-W_i) ]$. Therefore,
\begin{align*}
&\sqrt{n} \left[ \tilde G(w, x; 1) - \tilde G(w, x; 0) \right]\\
&=\sqrt{n} \left[ \mathbb U_1(w, x; 1) - \mathbb U_1(w, x; 0) \right]+\sqrt{n} \left[ \mathbb U_2(w, x; 1) - \mathbb U_2(w, x; 0) \right]\\
&=\sqrt{n} \left\{ \mathbb U_1(w, x; 1) - \mathbb U_1(w, x; 0)-\left[\mathbb E \mathbb U_1(w,x;1)-\mathbb E \mathbb U_1(w,x;0)\right]\right\}\\
&\quad+\sqrt{n} \left\{ \mathbb U_2(w, x; 1) - \mathbb U_2(w, x; 0)-\left[\mathbb E \mathbb U_2(w,x;1)-\mathbb E \mathbb U_2(w,x;0)\right]\right\}\\
&\quad+\sqrt n\left[\mathbb E \mathbb U_1(w,x;1)-\mathbb E \mathbb U_1(w,x;0)\right]+\sqrt n\left[\mathbb E \mathbb U_2(w,x;1)-\mathbb E \mathbb U_2(w,x;0)\right].
\end{align*}

We first look at those terms with $\mathbb U_2$.  By definition, 
\begin{align*}
&\mathbb U_2(w,x;z)\\
=&\frac{1}{n(n-1)}\sum_{i=1}^n\sum_{j\neq i}   \{ \Indi^{*}_{X_iZ_i}(x,z) \lambda(W_i-w) \times \frac{1}{h_q} K_q(\frac{X_j-X_i}{h_q})\mathbb 1(Z_j=1-z) \} \\
=&  \frac{1}{n(n-1)}\sum_{i=1}^n\sum_{j\neq i} \zeta_{n,ij}(w,w,z_\ell)
\end{align*} 
where $\zeta_{n,ij}(w,x,z)=\Indi^{*}_{X_iZ_i}(x,z) \lambda(W_i-w) \times \frac{1}{h_q} K_q(\frac{X_j-X_i}{h_q})\mathbb 1(Z_j=1-z)$.

Let $ \zeta^*_{n,ij}(w,x,z)=\frac{1}{2}\left[ \zeta_{n,ij}(w,x,z)+ \zeta_{n,ji}(w,x,z)\right]$. 
Then, $ \zeta^*_{n,ij}$ is symmetric in indices $i$ and $j$. 
Therefore,
\[
\mathbb U_{2}(w,x,z)= \frac{1}{n(n-1)}\sum_{i=1}^n\sum_{j\neq i}\zeta^*_{n,ij}(w,x,z),
\]
which is a $\mathcal{U}$-process indexed by $(w,x,z_\ell)$.  By \citet[Theorem 5]{nolan1988functional} and \citet[Lemma 3.1]{powell1989semiparametric},
\begin{eqnarray*}
&&\mathbb U_{2}(w,x, z)-\mathbb E\mathbb U_2(w,x,z)\\
&=& \frac{2}{n}\sum_{i=1}^n\left\{ \mathbb E[\zeta^*_{n,ij}(w,x,z)|Y_i,D_i,X_i,Z_i]- \mathbb E[\zeta^*_{n,ij}(w,x,z)]\right\}+ o_p(n^{-1/2}).
\end{eqnarray*}where the $o_p(n^{-1/2})$ applies uniformly over $(w,x)$.  Note that 
\begin{multline*}
\mathbb E[\zeta^*_{n,ij}(w,x,z)|Y_i,D_i,X_i,Z_i]\\
=\frac{1}{2}\Big\{\mathbb{1}^*_{XZ}(x,z_\ell) q(X,1-z) \lambda(W-w)
+\mathbb{1}^*_{XZ}(x,1-z) q(X,z) \Pi (w|X,z)\Big\}+o_p(1).
\end{multline*}


We now derive $\mathbb E [\zeta^*_{n,ij}(w,x,z)]$. 
Let $\mu_1(w,x,z)=\mathbb E [\mathbb{1}^*_{XZ}(x,z) q(X,1-z) \lambda(W-w)]$ and $\mu_2(w, x,z)=\mathbb E [\mathbb{1}^*_{XZ}(x,1-z) q(X,z) \Pi (w|X,z)]$.
Note that 
\[
\mu_1(w,x,z)=\mu_2(w,x,z)=\int \mathbb 1 (X\leq x) \Pi(w|X)f_{X|Z}(X|1)f_{X|Z}(X|0)d X\times \mathbb P(Z=1)\mathbb P(Z=0)
\] 
under the $\mathbb H^G_0$, which are invariant with $z$. 
Therefore, $\mathbb E [\zeta^*_{n,ij}(w,x,z)]=\frac{1}{2}[\mu_1(w,x,z)+\mu_2(w,x,z)]$ is also invariant with $z$. Let $\mu^*(w,x)=\mathbb E [\zeta^*_{n,ij}(w,x,z)]$.

By \citet[Theorem 3.1]{powell1989semiparametric}, 
\begin{align*}
&\frac{1}{\sqrt n}\sum_{i=1}^n\left\{ \mathbb E[\zeta^*_{n,ij}(w,x,z_\ell)|Y_i,D_i,X_i]- \mathbb E[\zeta^*_{n,ij}(w,x,z_\ell)]\right\} \\
&= \mathbb E_n\left\{\mathbb{1}^*_{XZ}(x,z_{\ell}) q(X,z_{-\ell}) \lambda(W-w)-\mu^*(w,x) \right\}\\
&\quad+\mathbb E_n\left\{\mathbb{1}^*_{XZ}(x,z_{-\ell}) q(X,z_{\ell}) \Pi(w|X,z_\ell)-\mu^*(w,x) \right\}+ o_p(n^{-\frac{1}{2}}), 
\end{align*}
where the $o_p(n^{-1/2})$ holds uniformly over $(w,x)$. It follows that 
\begin{align*}
&\mathbb U_2(w, x; 1) - \mathbb U_2(w, x; 0)-\left[\mathbb E \mathbb U_2(w,x;1)-\mathbb E \mathbb U_2(w,x;0)\right] 
=  \mathbb E_n\psi^c_{wx} +o_p(n^{-\frac{1}{2}}).
\end{align*}

We now turn to $\mathbb U_1(w,x,z)$. 
Note that
\[
\mathbb U_1(w,x;z) = -\frac{1}{n}\sum_{i=1}^n \left \{ \Indi^{*}_{W_iX_iZ_i}(w,x,z)  q(X_i,1-z) (1-D_i)  [\hat {\delta}(X_i)- \delta(X_i)] \right\} +o_p(n^{-\frac{1}{2}}),
\]
provided that  $\sup_{x\in\mathscr S_X}\left|\left[\hat q(x,z)- q(x,z)\right]\times \left[\hat \delta(x) -\delta(x)\right]\right|=o_p(n^{-\frac{1}{2}})$ holds.   
By a similar decomposition argument on $\hat \delta (X)-\delta(X)$ in \Cref{lem:smalldistance1}, we have  
\begin{align*}
\mathbb U_1(w,x;z)
&=-\frac{1}{n(n-1)}\sum_{i=1}^n\sum_{j\neq i} \Bigg \{ \Indi^{*}_{W_iX_iZ_i}(w,x,z)  q(X_i,1-z) (1-D_i) \\
&\qquad\times \frac{\overline W_{ji}\frac{1}{h}K(\frac{X_j-X_i}{h}) }{p(X_i,1)-p(X_i,0)}\left[\frac{ \mathbb 1(Z_j=1)}{q(X_i,1)}
-\frac{ \mathbb 1(Z_j=0)}{q(X_i,0)}\right]\Bigg\}+o_p(n^{-1/2})\\
&=-\frac{1}{n(n-1)}\sum_{i=1}^n\sum_{j\neq i}  \xi_{n,ij}(w,x,z_\ell)+o_p(n^{-1/2})
\end{align*} 
where $\overline W_{ji}= W_j-\mathbb E (W_j|X_i)$ and 
\begin{align*}
\xi_{n,ij}(w,x,z) &= \Indi^{*}_{W_iX_iZ_i}(w,x,z)  q(X_i,1-z) (1-D_i) \times \\
&\quad \frac{\overline W_{ji}\frac{1}{h}K(\frac{X_j-X_i}{h}) }{p(X_i,1)-p(X_i,0)}\left[\frac{ \mathbb 1(Z_j=1)}{q(X_i,1)}
-\frac{ \mathbb 1(Z_j=0)}{q(X_i,0)}\right].
\end{align*}
Moreover,  Let $ \xi^*_{n,ij}(w,x,z)=\frac{1}{2}[ \xi_{n,ij}(w,x,z)+ \xi_{n,ji}(w,x,z)]$.  
By a similar argument as that for $\mathbb U_2$, we have 
\begin{multline*}
\mathbb U_{1}(w,x,z) - \mathbb E \mathbb U_{1}(w,x,z)\\
= -\frac{2}{n}\sum_{i=1}^n\left\{ \mathbb E[\xi^*_{n,ij}(w,x,z)|Y_i,D_i,X_i,Z_i]- \mathbb E[\xi^*_{n,ij}(w,x,z)]\right\}+ o_p(n^{-1/2}).
\end{multline*}
Because
\begin{align*}
&\Exp [\xi_{n,ji}(w,x,z)|Y_i,D_i,X_i,Z_i] \\
&=\Exp \left\{\mathbb E [\xi_{n,ji}(w,x,z)|X_j,Z_j,Y_i,D_i,X_i,Z_i]\big|Y_i,D_i,X_i,Z_i\right\}\\
&= \Exp \Bigg\{\Indi^{*}_{X_jZ_j}(x,z)  q(X_j,1-z) \Pr(W\leq w;D=0|X_j,Z_j)[W_i-\mathbb E(W|X_j)] \\
&\quad\times  \frac{\frac{1}{h}K(\frac{X_i-X_j}{h}) }{p(X_j,1)-p(X_j,0)}\left[\frac{ \mathbb 1(Z_i=1)}{q(X_j,1)}-\frac{ \mathbb 1(Z_i=0)}{q(X_j,0)}\right]\Big| Y_i,D_i,X_i,Z_i\Bigg\}\\
&= \1 (X_i\leq x)  q(X_i,1-z) \Pr(W\leq w;D=0|X_i,Z_i=z) [W_i-\mathbb E(W|X_i)]\\
&\quad\times  \frac{q(X_i,z)}{p(X_i,1)-p(X_i,0)}\left[\frac{ \1 (Z_i=1)}{q(X_i,1)}-\frac{ \mathbb 1(Z_i=0)}{q(X_i,0)}\right]+o_p(1)
\end{align*}
where the last step comes from the Bochner's Lemma and uses the fact the integrant equals zero if $Z_j=1-z$, and the expectation of the first term in the above equation equals to zero. 


Thus, we have
\begin{multline*}
\mathbb U_{1}(w,x, z)-\mathbb E \mathbb U_{1}(w,x, z)\\
=- \mathbb E_n\Bigg\{    \overline W\times \frac{F_{WD|XZ}(w,0|X,z)}{p(X,1)-p(X,0)}\times \left[\frac{ \mathbb 1^*_{XZ}(x,1)}{q(X,1)}-\frac{ \mathbb 1^*_{XZ}(x,0)}{q(X,0)}\right] \\
\times q(X,1)\times q(X,0)\Bigg\} + o_p(n^{-\frac{1}{2}}), 
\end{multline*}where the $o_p(n^{-1/2})$ holds uniformly over $(w,x)$. 
Moreover,
\[
\mathbb U_{1}(w,x, 1)-\mathbb E \mathbb U_{1}(w,x, 1)-[\mathbb U_{1}(w,x, 0)-\mathbb E \mathbb U_{1}(w,x, 0)]= \mathbb E_n\phi^c_{wx} + o_p(n^{-\frac{1}{2}}).
\]

Finally, by \Cref{assmp:bias},
\[
\mathbb E \mathbb U_2(w,x;z)=\int_{-\infty}^x q(X,1)q(X,0) \Pi(w|X,Z=z) dX+ o_p(n^{-\frac{1}{2}})
\]which is invariant with $z$ under $\mathbb H^G_0$, 
and
\[
\mathbb E \mathbb U_1(w,x;z)=o_p(n^{-\frac{1}{2}}).
\]
Therefore, 
\[
\sqrt n\left[\mathbb E \mathbb U_1(w,x;1)-\mathbb E \mathbb U_1(w,x;0)\right]+\sqrt n\left[\mathbb E \mathbb U_2(w,x;1)-\mathbb E \mathbb U_2(w,x;0)\right]=o_p(n^{-\frac{1}{2}}).
\]

To conclude, the empirical process $\sqrt{n} [\tilde{G}(w,x;1) - \tilde{G}(w,x;0)]$ is asymptotically equivalent to the following process
\[
\sqrt n \times \mathbb E_n (\psi^c_{wx}+\phi^c_{wx}).
\]which converges to a zero-mean Gaussian process $\mathcal{Z}(\cdot,x)$ with the given covariance kernel.
\end{proof}

\section{Technical Lemmas}
\begin{lemma}
\label{lem:smalldistance1}
Suppose \Cref{assmp:densitysmoothness,assmp:bandwidth,assmp:kernel} hold. Then for any $k>0$ and $r\in(\frac{1}{4},\iota)$, 
\[
\sup_{x\in\mathscr S_X}n^k \times \Pr\left[|\hat\delta(x)-\delta(x)|> n^{-r}\right]\rightarrow 0.
\]
\end{lemma}
\begin{proof}
First, by a similar decomposition of $\hat \delta(x)-\delta(x)$ as that in the proof of \Cref{thm:test_noX}, it suffices to show 
\begin{align*}
&\sup_x n^k \times \Pr\left\{\left|a_{n}(x,z) - a(x,z)\right|>\lambda_{a}\times n^{-r}\right\}\rightarrow 0;\\
&\sup_xn^k \times \Pr\left\{\left|b_{n}(x,z) - b(x,z)\right|>\lambda_{b}\times n^{-r}\right\}\rightarrow 0;\\
&\sup_xn^k \times \Pr\left\{\left|q_{n}(x,z) - q(x,z)\right|>\lambda_{q}\times n^{-r}\right\}\rightarrow 0,
\end{align*} 
where  $\lambda_{a}$, $\lambda_{b}$ and $\lambda_{q}$ are strictly positive constants, and 
\begin{align*}
&a_{n}(x,z)= \frac{1}{nh}\sum_{j=1}^n Y_j K(\frac{X_j-x}{h}) \mathbb 1(Z_j=z),  \ \ \ a(x,z)=\mathbb E(Y|X=x,Z=z) \times q(x,z);\\
&b_{n}(x,z)= \frac{1}{nh}\sum_{j=1}^n D_j K(\frac{X_j-x}{h}) \mathbb 1(Z_j=z), \ \ \ b(x,z)=\mathbb E(D|X=x,Z=z) \times q(x,z);\\
&q_{n}(x,z)= \frac{1}{nh}\sum_{j=1}^n K(\frac{X_j-x}{h}) \mathbb 1(Z_j=z).
\end{align*}
For expositional simplicity, we only show the first result. It is straightforward that the rest follow a similar argument.

Let $T_{nxzj}=Y_jK(\frac{X_j-x}{h})\mathbb 1(Z_j=z)$ and $\tau_{nxz}= h\times \left[\lambda_{a}  n^{-r}-|\mathbb{E}a_{n}(x,z)-a(x,z)|\right]$.  Note that 
\begin{eqnarray*}
&&\Pr\left[|a_{n}(x,z)-a(x,z)|> \lambda_{a} \times n^{-r}\right]\\
&\leq& \Pr\left[| a_{n}(x,z)-\mathbb{E}a_{n}(x,z)|+|\mathbb{E}a_{n}(x,z)-a(x,z)|> \lambda_{a} \times n^{-r}\right]\\
&=&\Pr\left\{\frac{1}{n}\left|\sum_{j=1}^n\left( T_{nxzj}-\mathbb{E}T_{nxzj}\right)\right|> \tau_{nxz}\right\}.
\end{eqnarray*}
Moreover, by Bernstein's tail inequality,  
\[
\Pr\left\{\frac{1}{n}\left|\sumi \left( T_{xzj}-\mathbb{E} T_{xzj} \right)\right|> \tau_{nxz}\right\} 
\leq 2\exp\left(-\frac{n \times \tau_{nxz}^2}{2\text{Var}\left( T_{nxzj} \right)+\frac{2}{3}\overline K\times \tau_{nxz}}\right).
\]
where $\overline K$ is the upper bound of kernel $K$.

By \Cref{assmp:kernel}, $|\mathbb{E}a_{n}(x,z)-a(x,z)|=O(n^{-\iota})=o(n^{-r})$. Then, for sufficient large $n$, there is $ 0.5\lambda_{a} n^{-r}h\leq \tau_n(x,z)\leq \lambda_{a} n^{-r}h$. Moreover,
\[
\text{Var} \left( T_{nxzj} \right) \leq \mathbb{E} T^2_{nxzj} \leq  \mathbb{E} \big[\mathbb E (Y^2|X)K^2(\frac{X-x}{h})\big]\leq Ch,
\] 
where $C=\sup_{x} \mathbb E [Y^2|X=x]\times \sup_x f_X (x)\times \overline K \times  \int |K(u)| du<\infty$. It follows that
\[
\Pr\left\{\frac{1}{n}\left|\sum_{\ell=1}^n\left( T_{xzj}-\mathbb{E} T_{xzj}\right)\right|>\tau_{nxz}\right\} \leq 2\exp\left(-\frac{\frac{\lambda_{a}}{4}nhn^{-2r}}{2C +\frac{2}{3}\overline K \lambda_{a} n^{-r}}\right).
\]
For sufficiently large $n$, we have $\frac{2}{3}\overline K \lambda_{a} n^{-r}\leq 1$.  Therefore,  for sufficiently large $n$,
\[
\Pr\left\{\frac{1}{n}\Big|\sum_{\ell=1}^n\left(T_{xzj}-\mathbb{E}T_{xzj}\right)\Big|>\tau_{nxz}\right\} \leq 2 \exp\left(-\frac{  n^{2\iota-2r}}{2C +1}\right) =o(n^{-k})
\]where the inequality comes from \Cref{assmp:bandwidth}.  Note that the upper bound does not depend on $x$ or $z$. Therefore,
\[
\sup_{x,z}\Pr\left[|a_{n}(x,z)-a(x,z)|> \lambda_{a} \times n^{-r}\right]=o(n^{-k}).
\]
\end{proof}
\end{appendices}

\newpage
\bibliography{additivity_TC.bib}
\bibliographystyle{econometrica}

\newpage
\section*{Tables}
\begin{table}[!h]
\caption{Empirical level for DGP 1}  \label{table:dsize}
\begin{center}
\begin{tabular}{c | c | ccc | ccc | ccc}
 &  $\alpha$ & \multicolumn{3}{c|}{$0.01$} & \multicolumn{3}{c|}{$0.05$} &\multicolumn{3}{c}{$0.10$} \\
\hline
$N$ & $p\ \backslash\ \rho$ & $0.5$ & $0.7$ & $0.9$  & $0.5$ & $0.7$ & $0.9$  & $0.5$ & $0.7$ & $0.9$ \\
\hline
 \multirow{3}{*}{$500$} 
& $0.3$  &  $0.0010$  &   $0.0010$   &  $0.0010$ &  $0.0190$ & $0.0210$ & $0.0210$ & $0.0590$  &   $0.0640$  &   $0.0640$ \\ 
& $0.5$ & $0.0050$ & $0.0040$ & $0.0020$ & $0.0270$ & $0.0290$ & $0.0230$ & $0.0570$ & $0.0610$ & $0.0590$ \\ 
& $0.7$ & $0.0010$ & $0.0020$ & $0.0010$ & $0.0200$ & $0.0220$ & $0.0190$ & $0.0560$ & $0.0590$ & $0.0530$  \\ 
\hline
 \multirow{3}{*}{$1000$} 
& $0.3$ & $0.0040$ & $0.0080$ & $0.0060$ & $0.0450$ & $0.0390$ & $0.0360$ & $0.0810$ & $0.0730$ & $0.0710$ \\ 
& $0.5$ & $0.0030$ & $0.0030$ & $0.0050$ & $0.0370$ & $0.0360$ & $0.0290$ & $0.0740$ & $0.0700$ & $0.0750$\\ 
& $0.7$ & $0.0050$ & $0.0040$ & $0.0020$ & $0.0430$ & $0.0340$ & $0.0320$ & $0.0780$ & $0.0740$ &  $0.0750$ \\ 
\hline
 \multirow{3}{*}{$2000$} 
& $0.3$ & $0.0050$ & $0.0030$ & $0.0030$ & $0.0390$ & $0.0340$ & $0.0360$ & $0.0970$ & $0.0970$ & $0.0990$\\ 
& $0.5$ & $0.0060$ & $0.0070$ & $0.0080$ & $0.0480$ & $0.0470$ & $0.0470$ & $0.0930$ & $0.0950$ & $0.0930$\\ 
& $0.7$ & $0.0080$ & $0.0100$ & $0.0100$ & $0.0450$ & $0.0460$ & $0.0430$ & $0.0810$ & $0.0800$ & $0.0780$
\end{tabular}
\end{center}
\end{table}

\begin{table}[!h]
\caption{Empirical power for DGP 2: $\gamma = 0.1$}
\label{table:dpower_gamma_1}
\begin{center}
\begin{tabular}{c | c | ccc | ccc | ccc}
 &  $\alpha$ & \multicolumn{3}{c|}{$0.01$} & \multicolumn{3}{c|}{$0.05$} &\multicolumn{3}{c}{$0.10$} \\
\hline
$N$ &  $p\ \backslash\ \rho$ & $0.5$ & $0.7$ & $0.9$  & $0.5$ & $0.7$ & $0.9$  & $0.5$ & $0.7$ & $0.9$ \\
\hline
 \multirow{3}{*}{$500$} 
& $0.3$ & $0.0040$ &  $0.0030$ & $0.0070$ & $0.0360$ &$0.0330$ & $0.0300$ & $0.0880$ & $0.0820$  &   $0.0860$\\ 
& $0.5$ & $0.0100$ & $0.0090$ & $0.0070$ & $0.0380$ & $0.0450$ & $0.0370$ & $0.1050$ & $0.0950$  & $0.0900$\\ 
& $0.7$ & $0.0060$ & $0.0050$ & $0.0020$ & $0.0300$ & $0.0380$ & $0.0330$ & $0.0900$ & $0.0940$ & $0.0800$\\ 
\hline
 \multirow{3}{*}{$1000$} 
& $0.3$ & 0.0170  &   0.0170  &   0.0160 & 0.0930  &   0.0900 &   0.0840 & 0.1670 &   0.1580    & 0.1640 \\ 
& $0.5$ & 0.0210  &   0.0210  &   0.0210 & 0.1050  &   0.0950 &   0.0950 & 0.1650 &   0.1590    & 0.1410\\ 
& $0.7$ & 0.0160  &  0.0130   &  0.0120  & 0.0800  &  0.0780  &   0.0730 & 0.1650  &  0.1590    & 0.1410\\ 
\hline
 \multirow{3}{*}{$2000$} 
& $0.3$ & 0.0630  &   0.0650  &   0.0630  & 0.2260  &   0.1900  &   0.1940 & 0.3500  &   0.3160    & 0.3220\\ 
& $0.5$ & 0.0820  &   0.0690  &   0.0670  & 0.2500  &   0.2330  &   0.2340 & 0.3910  &   0.3710    & 0.3750\\ 
& $0.7$ & 0.0440  &   0.0450  &   0.0400  & 0.1830  &   0.1730  &   0.1740 & 0.3190  &   0.3020    & 0.2970
\end{tabular}
\end{center}
\end{table}

\begin{table}[!h]
\caption{Empirical power for DGP 2: $\gamma = 0.3$}
\label{table:dpower_gamma_3}
\begin{center}
\begin{tabular}{c|c|ccc|ccc|ccc}
 &  $\alpha$ & \multicolumn{3}{c|}{$0.01$} & \multicolumn{3}{c|}{$0.05$} &\multicolumn{3}{c}{$0.1$} \\
\hline
$N$ & $p\ \backslash\ \rho$ & $0.5$ & $0.7$ & $0.9$  & $0.5$ & $0.7$ & $0.9$  & $0.5$ & $0.7$ & $0.9$ \\
\hline
 \multirow{3}{*}{$500$} 
& $0.3$ &  0.0440  &   0.0370   &  0.0350  & 0.2070 & 0.1750  &  0.1820  &  0.3340   &   0.3120   & 0.3130\\ 
& $0.5$ &  0.0520  &   0.0590  &   0.0550 & 0.2380  & 0.2220  &  0.2340  &  0.4190   &   0.3870   & 0.3900\\ 
& $0.7$ & 0.0270  &   0.0230  &   0.0200  &  0.1630 & 0.1470  &  0.1390  &  0.3300   &  0.2980    & 0.2950\\ 
\hline
 \multirow{3}{*}{$1000$} 
& $0.3$ & 0.2610  &   0.1940   &  0.2120  &  0.5850  &   0.5140  &  0.5250  &  0.7460  &   0.6830   & 0.7030\\ 
& $0.5$ & 0.3430  &   0.3000  &   0.3040  &  0.7150   &  0.6390  &   0.6620  &  0.8690   &  0.8090    & 0.8260\\ 
& $0.7$ & 0.1950  &   0.1760   &  0.1870  &  0.5690    &   0.5080   &   0.5320 & 0.7770  &  0.7050 & 0.7180\\ 
\hline
 \multirow{3}{*}{$2000$} 
& $0.3$ &  0.8620  &   0.7620  &   0.7750  &  0.9850  &   0.9540  &   0.9710 & 0.9950 &    0.9870    & 0.9920\\ 
& $0.5$ &  0.9310  &   0.8810  &   0.8930  &  0.9990  &   0.9930  &   0.9940 & 1.0000 &    1.0000    & 1.0000\\ 
& $0.7$ &  0.7970  &   0.7030  &   0.7100  &  0.9770  &   0.9560  &   0.9550 & 0.9970 &    0.9850    & 0.9840
\end{tabular}
\end{center}
\end{table}

\begin{table}[!h]
\caption{Empirical power for DGP 2: $\gamma = 0.5$}
\label{table:dpower_gamma_5}
\begin{center}
\begin{tabular}{c|c|ccc|ccc|ccc}
 &  $\alpha$ & \multicolumn{3}{c|}{$0.01$} & \multicolumn{3}{c|}{$0.05$} &\multicolumn{3}{c}{$0.1$} \\
\hline
$N$ & $p\ \backslash\ \rho$ & $0.5$ & $0.7$ & $0.9$  & $0.5$ & $0.7$ & $0.9$  & $0.5$ & $0.7$ & $0.9$ \\
\hline
 \multirow{3}{*}{$500$} 
& $0.3$ &  0.1800   &  0.1310  &   0.1420 & 0.5120  &   0.4550  &   0.4640 & 0.7050  &  0.6540  &   0.6630\\ 
& $0.5$ &  0.2350   &  0.1870  &   0.2150 & 0.6840  &   0.6200  &   0.6180 & 0.8540  &   0.8170  &   0.8150\\ 
& $0.7$ & 0.1190    &  0.0860  &   0.0920 & 0.4630  &   0.3950  &   0.4390 & 0.6920  &   0.6440  &   0.6850\\ 
\hline
 \multirow{3}{*}{$1000$} 
& $0.3$ & 0.7950  &   0.7190  &   0.7320  & 0.9760 &    0.9540  &   0.9520 & 0.9950  &   0.9880    & 0.9880\\ 
& $0.5$ & 0.9110  &   0.8450  &   0.8540  & 0.9980 &    0.9850  &   0.9920 & 0.9990  &   0.9970    & 0.9990\\ 
& $0.7$ & 0.7240  &   0.6440  &   0.6480  & 0.9740 &    0.9520  &   0.9530 & 0.9980  &   0.9850    & 0.9910\\ 
\hline
 \multirow{3}{*}{$2000$} 
& $0.3$ & 1.0000 &    0.9980  &   0.9990 & 1.0000 &    1.0000  &   1.0000 & 1.0000  &   1.0000    & 1.0000\\ 
& $0.5$ & 1.0000 &    1.0000  &   1.0000 & 1.0000 &    1.0000  &   1.0000 & 1.0000  &   1.0000    & 1.0000\\ 
& $0.7$ & 1.0000 &    1.0000  &   1.0000 & 1.0000 &    1.0000  &   1.0000 & 1.0000  &   1.0000    & 1.0000
\end{tabular}
\end{center}
\end{table}

\begin{table}
\caption{Empirical size for DGP 3: $(p,\rho)=(0.5,0.7)$}
\label{table:csize}
\begin{center}
\begin{tabular}{c|ccc}
$N\ \backslash\ \alpha$ & $0.01$ & $0.05$ & $0.1$ \\
\hline
$1000$ & $0.0040$ & $0.0300$ & $0.0480$ \\
$2000$ & $0.0080$ & $0.0500$ & $0.0880$ \\
$4000$ & $0.0080$ & $0.0360$ & $0.0780$ 
\end{tabular}
\end{center}
\end{table}

\begin{table}[!h]
\caption{Empirical power for DGP 4: $(p,\rho)=(0.5,0.7)$}
\label{table:cpower}
\begin{center}
\begin{tabular}{c|ccc}
$N\ \backslash\ \alpha$ & $0.01$ & $0.05$ & $0.1$ \\
\hline
& \multicolumn{3}{c}{\underline{$\gamma = 0.5$}} \\
$1000$ & $0.0040$ & $0.0280$ & $0.0520$ \\
$2000$ & $0.0080$ & $0.0580$ & $0.1100$ \\
$4000$ & $0.1060$ & $0.3350$ & $0.5180$ \\
& \multicolumn{3}{c}{\underline{$\gamma = 0.7$}} \\
$1000$ & $0.0040$ & $0.0340$ & $0.0620$ \\
$2000$ & $0.0044$ & $0.1880$ & $0.3260$ \\
$4000$ & $0.7240$ & $0.9400$ & $0.9840$
\end{tabular}
\end{center}
\end{table}

\begin{table}[!h]
\caption{Descriptive Statistics for the National JTPA Study}
\label{table:jtpa}
\begin{center}
\begin{tabular}{lccc}
& \multirow{2}{*}{All} & $Z=1$ & $Z=0$ \\
&     & (eligible) & (not eligible) \\
\hline
Men\\
\hspace{4pt} Number of observations  & $5,102$   & $3,399$   & $1,703$ \\
\hspace{4pt} Training ($D=1$)        & $41.87\%$  & $62.28\%$   & $1.12\%$ \\
\hspace{4pt} High school or GED      & $69.32\%$  & $69.26\%$  & $69.43\%$ \\
\hspace{4pt} Married                 & $35.26\%$  & $36.01\%$  & $33.75\%$ \\
\hspace{4pt} Minorities              & $38.38\%$  & $38.69\%$  & $37.76\%$ \\
\hspace{4pt} Work less than 13 weeks in the past year  & $40.02\%$  & $40.28\%$ & $39.05\%$ \\
\hspace{4pt} 30 months earnings & $19,147$ & $19,520$ & $18,404$ \\
\\
Women\\
\hspace{4pt} Number of observations  & $6,102$   & $4,088$   & $2,014$ \\
\hspace{4pt} Training ($D=1$)        & $44.61\%$  & $65.73\%$  & $1.74\%$ \\
\hspace{4pt} High school or GED      & $72.06\%$  & $72.85\%$  & $70.45\%$ \\
\hspace{4pt} Married                 & $21.93\%$  & $22.48\%$  & $20.82\%$ \\
\hspace{4pt} Minorities              & $40.41\%$  & $40.58\%$  & $51.86\%$ \\
\hspace{4pt} Work less than 13 weeks in the past year  & $51.79\%$  & $51.75\%$ & $51.86\%$ \\
\hspace{4pt} 30 months earnings & $13,029$ & $13,439$ & $12,197$ \\
\hline
\end{tabular}
\begin{tablenotes}
\item[] Note: Means are reported in this table for the National JTPA study 30-month earnings sample. 
\end{tablenotes}
\end{center}
\end{table}

\begin{sidewaystable}
\begin{threeparttable}
\caption{Descriptive Statistics for the 1999 and 2000 Censuses}
\begin{center}
\label{table:fertility}
\begin{tabular}{@{\extracolsep{7pt}}lcccccc@{}}
& \multicolumn{3}{c}{1990} & \multicolumn{3}{c}{2000} \\
 \cline{2-4}  \cline{5-7} \\
& \multirow{2}{*}{All} & $Z=1$ & $Z=0$ & \multirow{2}{*}{All} & $Z=1$ & $Z=0$ \\
&     & (twin birth) & (no twin birth) &     & (twin birth) & (no twin birth) \\
\hline
Observations          & $602,767$ & $6,524$  & $596,243$ & $573,437$ & $8,569$  & $564,868$ \\
Number of children    & $1.9276$  & $2.5318$ & $1.9209$ & $1.8833$  & $2.5196$ & $1.8734$ \\
At least two children ($D=1$) & $0.6500$  & $1.0000$ & $0.6461$ & $0.6163$  & $1.0000$ & $0.6104$ \\
\\
Mother\\
\hspace{4pt} Age in years          & $29.7894$ & $29.9530$ & $29.7876$ & $30.0562$ & $30.3943$ & $30.0510$ \\
\hspace{4pt} Years of education    & $12.9196$ & $12.9623$ & $12.9191$ & $13.1131$ & $13.2615$ & $13.1108$ \\
\hspace{4pt} Black                 & $0.0637$  & $0.0757$  & $0.0636$ & $0.0724$  & $0.0816$  & $0.07228$ \\
\hspace{4pt} Asian                 & $0.0326$  & $0.0321$  & $0.0326$ & $0.0447$  & $0.0335$  & $0.0448$ \\
\hspace{4pt} Other Races           & $0.0537$  & $0.0592$  & $0.0536$ & $0.0912$  & $0.0806$  & $0.0914$ \\
\hspace{4pt} Currently at work     & $0.5781$  & $0.5444$  & $0.5785$ & $0.5629$  & $0.5132$  & $0.5637$ \\
\hspace{4pt} Usual hours per work  & $24.5660$ & $23.3537$ & $24.5795$ & $25.1400$ & $23.0491$ & $25.1723$ \\

\hspace{4pt} Wage or salary income last year & $8942$ & $8593$ & $8946$ & $14200$ & $13757$ & $14206$ \\

\\
Father \\
\hspace{4pt} Age in years          & $32.5358$ & $32.7534$ & $32.5333$ & $32.9291$ & $33.3102$ & $32.9232$ \\
\hspace{4pt} Years of education    & $13.0436$ & $13.0748$ & $13.0432$ & $13.0331$ & $13.1806$ & $13.0308$ \\
\hspace{4pt} Black                 & $0.0671$  & $0.0796$  & $0.0670$ & $0.0800$  & $0.0945$  & $0.0798$ \\
\hspace{4pt} Asian                 & $0.0291$  & $0.0263$  & $0.0292$  & $0.0402$  & $0.0318$  & $0.0403$ \\

\hspace{4pt} Other Races           & $0.0488$  & $0.0529$  & $0.0488$ & $0.0919$  & $0.0802$  & $0.0921$ \\
\hspace{4pt} Currently at work     & $0.8973$  & $0.8922$  & $0.8974$ & $0.8512$  & $0.8584$  & $0.8511$ \\
\hspace{4pt} Usual hours per work  & $42.7636$ & $42.7704$ & $42.7635$ & $43.8805$ & $43.8789$ & $43.8805$ \\
\hspace{4pt} Wage or salary income last year & $27020$ & $28039$ & $27010$ & $38041$ & $41584$ & $37987$ \\
\\
Parents \\
Wages or salary income last year & $35,963$ & $36,632$ & $35,956$ & $52,241$ & $55,342$ & $52,193$\\
\bottomrule
\end{tabular}
\begin{tablenotes}
\item[] Note: Data from the $1\%$ and $5\%$ PUMS in 1990 and 2000. Own calculations using the PUMS sample weights. The sample consists of married mother between 21 and 35 years of age with at least one child.
\end{tablenotes}
\end{center}
\end{threeparttable}
\end{sidewaystable}

\end{document}